\documentclass{article}
\usepackage[utf8]{inputenc}
\usepackage{upref}
\usepackage{amsmath}
\usepackage{todonotes}
\usepackage{amsthm}
\usepackage{comment}
\usepackage{amsfonts}
\usepackage{dsfont}
\usepackage{amssymb}
\usepackage{geometry}
\usepackage{esint}
\usepackage{bbm}
\usepackage{xcolor}
\usepackage{cancel}
\usepackage[colorlinks=true, linkcolor=black, citecolor=red] {hyperref}
\usepackage{appendix}

\usepackage[maxbibnames=5]{biblatex}
\newtheorem{theorem}{Theorem}[section]
\newtheorem{lemma}[theorem]{Lemma}

\newtheorem{corollary}[theorem]{Corollary}
\newtheorem{proposition}[theorem]{Proposition}

\newtheorem{definition}[theorem]{Definition}

\theoremstyle{remark}
\newtheorem{remark}[theorem]{Remark}

\numberwithin{equation}{section}

\DeclareMathOperator{\R}{\mathbb{R}}

\DeclareMathOperator{\N}{\mathbb{N}}

\DeclareMathOperator{\Hcal}{\mathcal{H}}
\DeclareMathOperator{\tr}{Tr}

\DeclareMathOperator{\Nup}{\mathit{N}_{\uparrow}}
\DeclareMathOperator{\Ndown}{\mathit{N}_{\downarrow}}
\DeclareMathOperator{\rup}{\rho_{\uparrow}}

\DeclareMathOperator{\rdown}{\rho_{\downarrow}}

\DeclareMathOperator{\supp}{supp}

\renewcommand{\d}{ \, \mathrm{d}}

\usepackage{biblatex} 
\begin{document}
	
	\title{Ground state energy of a dilute inhomogeneous Fermi gas}
	\author{Thomas Gamet\thanks{Ecole Normale Supérieure de Lyon, UMPA (UMR 5669)}}
	\maketitle
	
	\begin{abstract} We study the ground state energy of a system of $N$ fermions with two spin states in the large $N$ limit. The particles are placed in an inhomogeneous trapping pontential and interact via scaled interactions. We study a dilute limit where the range of the interaction potential is much smaller than the typical inter-particle distance. We show that the energy per particle converges to the Thomas-Fermi energy of the system, with a perturbative term corresponding to the interaction and exhibiting the scattering length of the potential. The proof is decomposed into two bounds. First, we construct an appropriate test-state to prove the upper bound. Then, we prove the lower bound by the Dyson lemma which allows us to regularize the interaction potential, and several semi-classical tools. 
	\end{abstract}
	
	\tableofcontents
	
	\section{Introduction}
	
	\subsection{Motivation and context}
	
	The asymptotics of the ground state energy per unit of volume of a dilute Fermi gas with short range interactions were first derived for the case of a hard-sphere potential in \cite{huang_quantum-mechanical_1957, lee_many-body_1957}: in three dimensions, in the thermodynamic limit and for two spins, for a total density of particles $\varrho$, the volumic energy $e$ follows the so called Huang-Yang formula in the low density limit
	\begin{equation}
		e(\varrho) = \frac{3}{5}(3\pi^2)^{2/3}\varrho^{5/3} + 2\pi a \varrho^2(1 + c\varrho^{1/3}) + o(\varrho^{7/3}),
	\end{equation}
	with $a$ the scattering length of the interaction and $c$ an explicit constant. The first order corresponds to the energy of a free gas, that is the kinetic energy of a system without interactions. The second order expansion was rigorously derived in \cite{lieb_ground-state_2005}. Several subsequent works have improved the result of \cite{lieb_ground-state_2005}. The method developped in \cite{falconi_dilute_2021} was also used in \cite{giacomelli_optimal_2023} and \cite{giacomelli_optimal_2024} to find optimal lower and upper bounds. More precisely, under a smoothness assumption on the interaction potential, the author derives the second term with a $O(\varrho^{7/3})$ error term. Without this smoothness assumption, \cite{lauritsen_almost_2025} gets a slightly worse error term of the form $O(\varrho^{7/3-\delta})$ for all $\delta>0$. More recently, the third term was recovered for the upper bound in \cite{giacomelli_huang-yang_2024} without any smoothness assumption, and in \cite{giacomelli_huang-yang_2025} for the lower bound. Other works including \cite{agerskov_ground_2025} focus on the one dimensional case.\\	
	In this paper, we will not consider a thermodynamic limit, but rather work with a trap of fixed size and study the limit when the number of particles $N$ tends to $+\infty$. Several papers compute the energy of a system of fermions on a torus, with a Hamiltonian of the form
	\begin{equation}
		\sum_{j=1}^N \hbar^2 (-\Delta_j) + \sum_{1\leq j < k\leq N}w_N(x_j - x_k)
	\end{equation}
	with $\hbar = N^{-1/3}$. The choice of $\hbar$ comes from the growth $\propto N^{5/3}$ of the sum of the $N$ first eigenvalues of the Laplacian in a fixed volume. For $N$ particles trapped in a space of size of order $1$, the typical distance between two particles is therefore of order $\hbar$. In \cite{fournais_ground_2024}, the strongly interacting case is studied, that is $w_N = N^{-\alpha}w$ with $\alpha < 1$. Another scaling for the interaction is the mean field regime where $w_N = N^{-1}w$. For this model, several papers have estimated the correlation energy, that is the difference between the ground state energy and the Hartree-Fock energy (the ground state energy restricted to Slater determinants, which is a first order estimate of the ground state energy). In \cite{hainzl_correlation_2020} and \cite{benedikter_optimal_2020}, the authors have to impose a weak interaction assumption on $w$. In \cite{benedikter_correlation_2021}, \cite{christiansen_random_2023} and \cite{benedikter_correlation_2023} however, this assumption is relaxed. Other scaling have been studied, such as $w_N = N^2w(N\cdot)$ in \cite{chen_second_2024}. This is a special case of the scaling $w_N = N^{3\beta - 1}w(N^{\beta}\cdot)$, which is natural as $\int w_N = \int w$. Note that when $\beta = 0$, we recover the mean field regime. When $\beta < 1/3$, since the typical distance between two fermions of same spin is $\hbar = N^{-1/3}$, particles of same spin interact, while when $\beta > 1/3$, the interaction between particles of same spin is way weaker than the interaction of particles of different spin because of the Pauli exclusion principle.\\	
	Let us now take a look at the pair interaction energy. For a given function $f$ (that corresponds to the distribution of the relative coordinate of a particle pair) on $\R^3$ such that $f(x) \to 1$ when $|x| \to +\infty$, the scattering energy is
	\begin{equation}
		\mathcal{E}^{\mathrm{s}}[f] = \int \Big(\hbar^2|\nabla f|^2 + \frac{1}{2}w_N(x) |f(x)|^2 \Big) \d x.
	\end{equation}
	With a change of variables of the form $y = N^{\beta}x$, and with $\tilde{f}(y) = f(N^{-\beta}y)$, we obtain 
	\begin{equation}
		\mathcal{E}^{\mathrm{s}}[f] = N^{-3\beta}\int \Big( N^{2\beta-2/3}|\nabla \tilde{f}(y)|^2 + N^{3\beta-1}|\tilde{f}(y)|^2 \Big)\d y.
	\end{equation}
	We immediately notice that the two terms in the integral are not of the same order, and this implies that the scaling $w_N =N^{3\beta-1}w(N^{\beta}\cdot)$ does not allow to recover the scattering length (more on this in Section \ref{subsection_diffusion_length}) of $w$ but only its range $R_w$ (defined by \eqref{equation_definition_range}). Therefore, in the following, we choose the scaling $w_N = N^{2\beta - 2/3}w(N^{\beta}\cdot)$.\\	
	Until now, all the models implied the confinement of the particles in a box, wether the size of the box was varying with the number of particles or not. However, models with non homogeneous trapping potentials have also been studied in the litterature, see for instance \cite{fournais_semi-classical_2018, lewin_semi-classical_2019}. We want to study such an inhomogeneous model in the dilute limit; namely in the case where $\beta > 1/3$, that is when interactions occur mainly between particles of different spin. More specifically, for $\hbar = N^{-1/3}$, we consider
	\begin{equation}\label{equation_definition_hamiltonien_intro}
		H_N = \sum_{j=1}^N (-\hbar^2\Delta_j) + \sum_{j= 1}^N V(x_j) + \sum_{1\leq j < k \leq N} w_N(x_j - x_k)
	\end{equation}
	with a given trapping potential $V$ that satisfies $V(x) \to + \infty$ when $|x| \to +\infty$ and a positive and compactly supported interaction potential 
	\begin{equation} 
		w_N =N^{2\beta-2/3}w(N^{\beta}\cdot). 
	\end{equation}
	Our goal is to estimate the ground state energy, i.e. the lowest eigenvalue of \eqref{equation_definition_hamiltonien_intro} on the appropriate Hilbert space for fermions with two spin states`. To do so, we will now define it properly in order to state a precise result.
	
	\subsection{Main result}
	First, let us define the space on which $H_N$ acts. Let 
	\begin{equation}
		\Hcal_{\Nup, \Ndown} = \bigg\{\Psi\in \Big(L^2_{\mathrm{as}}\big((\R^d)^{\Nup}\big)\otimes L^2_{\mathrm{as}}\big((\R^d)^{\Ndown}\big)\Big) \cap H^1\big((\R^d)^N\big),~ \int \Big(\sum_{j=1}^N V(x_j)\Big) |\Psi(X_N)|^2\d X_N< + \infty\bigg\}
	\end{equation}
	be the space of states with $\Nup$ spin-up and $\Ndown$ spin-down particles, with finite kinetic and potential energies. This corresponds to functions that are antisymmetric separately for the $\Nup$ first variables and for the $\Ndown$ last variables. We see $H_N$ as a quadratic form on $\Hcal_{\Nup, \Ndown}$ with $\Nup + \Ndown =N$, and we define the ground state energy by
	\begin{equation}
		E(N) = \inf_{N_{\uparrow} + N_{\downarrow} = N} \inf \left\{  \frac{\langle\Psi_{\Nup, \Ndown}|H_N| \Psi_{\Nup, \Ndown}\rangle}{\langle \Psi_{\Nup, \Ndown}|\Psi_{\Nup, \Ndown}\rangle},~ \Psi_{\Nup, \Ndown}\in \Hcal_{\Nup, \Ndown}\setminus
		\{0\}\right\}.
	\end{equation}
	
	\begin{definition}
		We define the scattering length $a_v$ of a radial positive and compactly supported potential $v$ on $\R^3$ by 
		\begin{equation}
			4\pi a_v = \inf\left\{\int |\nabla f|^2 + \frac{1}{2}v|f|^2,~ f\underset{\infty}{\to}1\right\}.
		\end{equation}
	\end{definition}	
	More information on the scattering length will be given in Section \ref{subsection_diffusion_length}. To estimate $E(N)$, we introduce the perturbed two-spin Thomas-Fermi energy functional 
	\begin{equation}\label{equation_perturbed_two_spin_thomas_fermi_functional}
		\mathcal{E}^{\mathrm{TF}}_N[\rup, \rdown] = c_{\mathrm{TF}}\int \Big(\rho_{\uparrow}^{5/3} + \rho_{\downarrow}^{5/3}\Big) + \int V(\rup + \rdown) + 8\pi a_w N^{1/3 - \beta}\int \rup \rdown,
	\end{equation}
	with the Thomas-Fermi constant
	\begin{equation} 
		c_{\mathrm{TF}} = \frac{3}{5}(6\pi^2)^{2/3} 
	\end{equation}
	and $a_w$ the scattering length of $w$. Moreover, we define the perturbed two-spin Thomas-Fermi energy by
	\begin{equation}\label{equation_definition_thomas_fermi_energy_deux_spin_perturbed}
		E^{\mathrm{TF}}_N = \inf\left\{ \mathcal{E}^{\mathrm{TF}}_N[\rup, \rdown],~ \rup, \rdown \geq 0,~ \int(\rup + \rdown) = 1\right\}.
	\end{equation}
	The kinetic term of $\mathcal{E}_N^{\mathrm{TF}}$ favors an equidistribution $\rup = \rdown$ of the density. We denote by $E^{\mathrm{TF}}$ the Thomas-Fermi energy defined by 
	\begin{equation}
		E^{\mathrm{TF}} = \inf\left\{ 2^{-2/3}c_{\mathrm{TF}} \int \rho^{5/3} + \int V\rho,~ \rho \geq 0,~ \int \rho = 1 \right\},
	\end{equation}  
	 and we denote by $\rho^{\mathrm{TF}}$ the minimizer. The proof of the existence of $\rho^{\mathrm{TF}}$ and more information will be given in Section \ref{subsection_Thomas_Fermi}. Before we state the theorem, we have to introduce some assumptions on the trapping potential $V$, that we call (\ref{H_1}) and (\ref{H_2}). 
	\begin{equation}\tag{$H_1$}\label{H_1}
		\begin{cases}
			V \in W^{3, \infty}_{\mathrm{loc}}\\
			V\geq 1\\
			|\Delta V| \leq CV^2\\
			|\nabla\Delta V| \leq CV\\
			\sum_{j,k}|\partial_{jk}V|^2 \leq CV^2
		\end{cases}
	\end{equation}
	
	\begin{equation}\label{H_2}\tag{$H_2$}
		\Lambda \mapsto \iint \mathds{1}_{|p|^2 + V(x) \leq \Lambda} \d x \d p~ \text{is differentiable}
	\end{equation}
	
	\begin{theorem}[Main result: convergence of the ground state energy]\label{theorem}
		For $1/3 < \beta < 34/81$, and under \emph{(}\ref{H_1}\emph{)} and \emph{(}\ref{H_2}\emph{)}, we have
		\begin{equation}
			E(N) = N E_N^{\mathrm{TF}} + o(N^{4/3 - \beta}) = N E^{TF} + 2\pi a_w N^{1/3-\beta} \int \big(\rho^{\mathrm{TF}}\big)^2 + o(N^{4/3-\beta}).
		\end{equation}
	\end{theorem}
	\begin{remark}[Comments on Theorem \ref{theorem}]~
		\begin{itemize}
			\item If we replace $w_N = N^{2\beta-2/3}w(N^{\beta}\cdot)$ by $N^{3\beta-1}w(N^{\beta}\cdot)$, Theorem \ref{theorem} holds true with the scattering length $a_w$ replaced by the range of the interaction $R_w$, defined by \eqref{equation_definition_range}. More detail on this in Section~\ref{subsection_diffusion_length}.
			\item The hypotheses (\ref{H_1}) and (\ref{H_2}) on $V$ are only used in the proof of the lower bound. It is clear that (\ref{H_1}) is staisfied for a trapping potentiel of the form $V(x) = 1 + |x|^s$ with $s>1$. 
			\item The hypothesis on $\beta$ is strong. It comes from the methods we used. In particular, for the lower bound, we compute separately the main term (composed of the potential and kinetic energies, of order $N$) and the interaction energy (of order $N^{4/3-\beta}$). However, the result should probably hold true for larger $\beta$.
		\end{itemize}        
		\hfill $\diamond$
	\end{remark}
	
	\subsection{Plan of the paper}
	
	Before we give an outline of the proof, let us explain why it is reasonable to expect Theorem \ref{theorem} to be true. First, let us focus on the kinetic and potential energies, i.e. let us consider $H_N^{w=0} = \sum_j \big((-\hbar^2 \Delta_j) + V(x_j)\big)$. For fermionic states, it is well known that $\sum_{j=1}^N \Delta_j \propto N^{5/3}$, since particles cannot occupy the same state, this is why we choose $\hbar = N^{-1/3}$, in order to have both the kinetic and potential energies of same order. For a state $\Psi$, one can expect the following semi-classical equality to hold
	\begin{equation}
		\langle\Psi|H_N^{w=0}|\Psi\rangle \approx N \iint \big(|p|^2 + V(x)\big)\big(m_{\uparrow}(x, p) + m_{\downarrow}(x, p)\big) \d x \d p,
	\end{equation}
	$m_{\uparrow}$ and $m_{\downarrow}$ denoting here the Husimi functions associated to $\Psi$. If $\Psi$ minimizes $H_N^{w = 0}$, then we can assume that $m_{\uparrow}$ is of the form $m_{\uparrow}(x, p) \approx \mathds{1}_{|p|\leq c \rup(x)}$ with $\rho(x) = \int m(x, p)\d p$, which gives
	\begin{equation}
		\langle\Psi|H_N^{w=0}|\Psi\rangle \approx c_{\mathrm{TF}}N\int \Big(\rho_{\uparrow}^{5/3} + \rho_{\downarrow}^{5/3}\Big) + N\int V\big(\rup + \rdown\big).
	\end{equation}
	It remains to explain the form of the interaction term. We won't explain where the exact constant comes from, but we explain why it is reasonable to expect something proportional to $N^{4/3 - \beta}R_w$ for the case $w_N = N^{3\beta-1}w(N^{\beta}\cdot)$ ($R_w$ being the range of $w$). As the range of the interaction is very short ($N^{-\beta} \ll \hbar$), it is natural to only see the interaction between spin-up and spin-down particles. Since the intensity of the interaction potential becomes very high, we can imagine that it is "almost" a hardcore potential, and that to minimize the energy, we have to prevent two particles to be at a distance lower than $N^{-\beta}R_w$. This implies an increase in the kinetic energy of order $N^{-2/3}$ (that corresponds to $\hbar^2$) times $N$ (typical number of spin-up particles) times $N$ (typical number of spin-down particles) times $N^{-\beta}R_w$ (typical length of the interaction), wich indeed gives $R_w N^{4/3 - \beta}$. When $w_N = N^{2\beta-2/3}w(N^{\beta}\cdot)$, it is reasonable to expect that $R_w$ is replaced by $a_w$.\\	
	Before moving on to the proof of Theorem \ref{theorem}, we establish some preliminary results in Section \ref{section_preliminaries}. We show that $E_N^{\mathrm{TF}}$ can be seen as the minimal energy of the functional without the interaction, plus an interaction term. The functional without the interaction is easier to study, and in particular, we can show that there exist minimizers and give some of their properties. We also discuss the notion of scattering length briefly and show that the scattering length of a potential whose intensity goes to infinity tends to the radius of the potential. Then, in Section \ref{section_upper_bound}, we prove the upper bound of the theorem. More precisely, we use the upper bound of \cite{lieb_ground-state_2005} and map our Hamiltonian to a context similar to that of \cite{lieb_ground-state_2005}. For their upper bound, Lieb, Seiringer and Solovej show that the thermodynamic setting can be replace by a system where the size of the box and the number of particles inside are a function of the density $\varrho$, which we can adapt to the case where $\varrho$ itself depends on our $N$. It is certainly possible to get the upper bound without using \cite{lieb_ground-state_2005} by constructing a test state directly; this could even give the bound with less restrictions on $\beta$. In Section \ref{section_lower_bound}, we show the lower bound of Theorem \ref{theorem}. For this bound, it is harder to use directly the lower bound of \cite{lieb_ground-state_2005}, because all the computations are made after taking the $L\to+\infty$ limit, while we would like to use it at large but finite $L$. Therefore, we use a more direct approach; we take a sequence of states with low energy and prove that its energy is bounded from below. We decompose the energy between the low frequency kinetic and potential energies that recover the unperturbed Thomas-Fermi energy on the one hand, and the high kinetic and interaction energies, that recover the perturbation, after a regularization by the Dyson lemma, on the other hand.
	
	\section{Preliminaries}\label{section_preliminaries}
	
	\subsection{Thomas Fermi functional and density}\label{subsection_Thomas_Fermi}
	
	Here, we only assume that $V\in L^{\infty}_{\mathrm{loc}}$ is a positive function such that $V(x) \to +\infty$ when $|x| \to +\infty$. In (\ref{equation_perturbed_two_spin_thomas_fermi_functional}), we have defined the perturbed two-spin Thomas-Fermi energy functional. Let us now introduce a similar functional without the interaction term. This functional has the same minimum as another functional involving only the total density, which will be easier to study. For a pair of densities $(\rho_{\uparrow}, \rho_{\downarrow})$, we define the two-spin Thomas-Fermi energy functional by
	\begin{equation}
		\mathcal{E}^{\mathrm{TF}}_{(2)}[\rho_{\uparrow}, \rho_{\downarrow}] = c_{\mathrm{TF}}\int \Big(\rho_{\uparrow}^{5/3} + \rho_{\downarrow}^{5/3}\Big) + \int V(\rho_{\uparrow} + \rho_{\downarrow}).
	\end{equation}
	Moreover, let $E^{\mathrm{TF}}$ denote the Thomas-Fermi energy, defined by
	\begin{equation}\label{definition_energy_thomas_fermi_}
		E^{\mathrm{TF}} = \inf \bigg\{ \mathcal{E}^{\mathrm{TF}}_{(2)}[\rho_{\uparrow}, \rho_{\downarrow}], ~\rho_{\uparrow}, \rho_{\downarrow} \geq 0, \int \rho_{\uparrow} + \int \rho_{\downarrow} = 1\bigg\}.
	\end{equation}
	Note that for a fixed total density $\rho_{\uparrow} + \rho_{\downarrow}$, $\mathcal{E}^{\mathrm{TF}}$ is minimized by $\rho_{\uparrow} = \rho_{\downarrow}$. We introduce then the Thomas-Fermi energy functional for a given total density $\rho$ by
	\begin{equation}
		\mathcal{E}^{\mathrm{TF}}[\rho] = 2^{-2/3}c_{\mathrm{TF}}\int \rho^{5/3} + \int V \rho,
	\end{equation}
	and we have 
	
	\begin{equation}
		E^{\mathrm{TF}} = \inf \bigg\{ \mathcal{E}^{\mathrm{TF}}[\rho], ~\rho \geq 0, \int \rho = 1\bigg\}.
	\end{equation}
	Note that the Thomas-Fermi energy is equal to the Vlasov energy that we now define. For a phase-space density $m$, we define the Vlasov energy functional by
	\begin{equation}
		\mathcal{E}^{\mathrm{V}}[m] = \frac{1}{(2\pi)^3}\int \big(|p|^2 + V(x) \big) m(x, p) \d x \d p.
	\end{equation}
	For a given $m$, we define the spatial density $\rho_m$ by
	\begin{equation}
		\rho_m(x) = \frac{1}{(2\pi)^3}\int m(x, p)\d p.
	\end{equation} 
	Then, by taking $\tilde{m}(x, p) = \mathds{1}_{|p|\leq c \rho_{m}(x)^{1/3}}$, we find
	\begin{equation}
		\mathcal{E}^{\mathrm{V}}[\tilde{m}] \leq \mathcal{E}^{\mathrm{V}}[m].
	\end{equation}
	Therefore, the Thomas Fermi energy is indeed equal to the Vlasov energy $E^{\mathrm{V}}$
	\begin{equation}
		E^{\mathrm{TF}} = E^{\mathrm{V}} := \inf\bigg\{ \mathcal{E}^{\mathrm{V}}[m],~  0 \leq m \leq 2,~ \int m = (2\pi)^3 \bigg\}.
	\end{equation}
	In the following proposition, we recap the existence of a minimizer of the Thomas-Fermi functional, and state some of its properties.
	
	\begin{proposition}[Existence and properties of the minimizing density]\label{proposition_properties_thomas_fermi_density}~
		\begin{enumerate}
			\item There exists a unique minimizer of the Thomas-Fermi functional that we denote by $\rho^{\mathrm{TF}}$.
			\item There exists a Lagrange multiplyer $\lambda_{\mathrm{TF}}$ such that the following equality holds on $\supp \rho^{\mathrm{TF}}$
			\begin{equation}\label{equation_egalite_multiplicateur_lagrange}
				2^{-2/3}\frac{5}{3}c_{\mathrm{TF}}\big(\rho^{\mathrm{TF}}\big)^{2/3} + V = \lambda_{\mathrm{TF}},
			\end{equation}
			and the following inequality holds on $\big(\supp \rho^{\mathrm{TF}}\big)^c$
			\begin{equation}\label{equation_inegalite_multiplicateur_lagrange}
				V \geq \lambda_{\mathrm{TF}}.
			\end{equation}
			\item The minimizing density is compactly supported and satisfies $\rho^{\mathrm{TF}}\in L^{\infty}$. Moreover, if $V\in W^{1, \infty}_{\mathrm{loc}}$, $\rho^{\mathrm{TF}}\in W^{1, \infty}$.
		\end{enumerate}
	\end{proposition}
	\begin{proof}~\\
			Step $1.$ First, as the Thomas-Fermi functional is strictly convex, there is at most one minimizer. Then, it remains to show that there exists at last one minimizer. Since everything is positive, we have $E^{\mathrm{TF}} \geq 0$. Let $(\rho_n)$ be a minimizing sequence of $\mathcal{E}^{\mathrm{TF}}$. Then, $(\rho_n)$ is clearly bounded in $L^{5/3}$ and as $V(x) \to + \infty$ when $|x| \to +\infty$, $(\rho_n)$ is tight. Thus, up to a subsequence, $(\rho_n)$ converges weakly to a certain $\rho\in L^{5/3}$ satisfying $\int \rho = 1$. By lower semicontinuity of the $L^{5/3}$ norm, we have
			\begin{equation}
				\int \rho^{5/3} \leq \liminf \int \rho_n^{5/3}.
			\end{equation}
			Likewise, we have
			\begin{equation}
				\int V \rho \leq \liminf \int V \rho_n.
			\end{equation}
			Therefore, we finally have 
			\begin{equation}
				\mathcal{E}^{\mathrm{TF}}[\rho] \leq \liminf \mathcal{E}^{\mathrm{TF}}[\rho_n] = E^{\mathrm{TF}}.
			\end{equation}
			Step $2.$ We set $G = 2^{-2/3}\frac{5}{3}c_{\mathrm{TF}}\big(\rho^{\mathrm{TF}}\big)^{2/3} + V$. Let $\sigma \in C^{\infty}_c\big(\supp \rho^{\mathrm{TF}}\big)$ such that for a certain $\varepsilon_0> 0 $, $\rho^{\mathrm{TF}} + \varepsilon_0\sigma \geq 0$, and satisfying $\int \sigma = 0$. Then, for $0 <\varepsilon < \varepsilon_0$, we have
			\begin{equation}
				\mathcal{E}^{\mathrm{TF}}[\rho + \varepsilon \sigma] \geq \mathcal{E}^{\mathrm{TF}}[\rho].
			\end{equation}
			Taking $\varepsilon \to 0$, this implies
			\begin{equation}
				\int G \sigma = 0.
			\end{equation}
			Therefore, $G$ is constant on $\supp \rho^{\mathrm{TF}}$ and we call the constant $\lambda^{\mathrm{TF}}$. Now, let $\sigma \in C^{\infty}_c\big((\supp \rho^{\mathrm{TF}})^c\big)$ such that $\sigma \geq 0$ and $\int \sigma = 1$. Then, for $0 < \varepsilon < 1$, we have
			\begin{equation}
				\mathcal{E}^{\mathrm{TF}}\big[(1 - \varepsilon)\rho^{\mathrm{TF}} +\varepsilon \sigma\big] \geq \mathcal{E}^{\mathrm{TF}}[\rho^{\mathrm{TF}}].
			\end{equation}
			Taking $\varepsilon \to 0$, this implies
			\begin{equation}
				\int V \sigma = \int G \sigma \geq \int G \rho^{\mathrm{TF}} = \lambda^{\mathrm{TF}}.
			\end{equation}
			Therefore, on $(\supp \rho^{\mathrm{TF}})^c$, we have $V \geq \lambda^{\mathrm{TF}}$.\\~\\	
			Step $3.$ By (\ref{equation_egalite_multiplicateur_lagrange}), $V \leq \lambda^{\mathrm{TF}}$ on $\supp \rho^{\mathrm{TF}}$. Since $V(x) \to + \infty$ when $|x|\to + \infty$, $\supp \rho^{\mathrm{TF}}$ is bounded and $\rho^{\mathrm{TF}}$ is compactly supported. Moreover, as $V\in L^{\infty}_{\mathrm{loc}}$, $V$ is bounded on the support of $\rho^{\mathrm{TF}}$ and by (\ref{equation_egalite_multiplicateur_lagrange}), $\rho^{\mathrm{TF}}$ is bounded. When $V\in W^{1, \infty}_{\mathrm{loc}}$, it is clear that $\rho^{\mathrm{TF}}\in W^{1, \infty}$ by (\ref{equation_egalite_multiplicateur_lagrange}).
	\end{proof}
	Now that we have established the existence and properties of the minimizer of $\mathcal{E}^{\mathrm{TF}}$, we show that the minimizing sequences of the Thomas Fermi functionals converge strongly to the minimizer (we will use this crucially to prove Lemma \ref{lemma_convergence_pseudo_density} and Lemma \ref{lemma_convergence_densite_projecteur}).
	
	\begin{lemma}[Strong convergence of a TF-minimizing sequence]\label{lemme_convergence_suites_minimisantes}~
		\begin{enumerate} 
			\item Let $(\rho_n, \sigma_n)$ be a minimizing sequence of $\mathcal{E}_{(2)}^{\mathrm{TF}}$, that is $\mathcal{E}_{(2)}^{\mathrm{TF}}[\rho_n, \sigma_n] \to E^{\mathrm{TF}}$, then we have 
			\begin{equation}
				\rho_n,~\sigma_n \to \frac{1}{2}\rho^{\mathrm{TF}}
			\end{equation}
			strongly in $L^{5/3}$.
			\item Let $(\rho_n)$ be a minimizing sequence of $\mathcal{E}^{\text{TF}}$, that is $\mathcal{E}^{\mathrm{TF}}[\rho_n] \to E^{\mathrm{TF}}$, we have 
			\begin{equation}
				\rho_n \to \rho^{\mathrm{TF}}
			\end{equation}
			strongly in $L^{5/3}$.
		\end{enumerate}
	\end{lemma}
	\begin{proof}~\\
		Step 1. Let $(\rho_n, \sigma_n)$ such that $\int (\rho_n + \sigma_n) = 1$ and $\mathcal{E}_{(2)}^{\mathrm{TF}}[\rho_n, \sigma_n] \to E^{\mathrm{TF}}$. We have convergence up to extraction of $\rho_n$ weakly-$*$ in $(L^{\infty})^*$ and weakly in $L^{5/3}$, to a $\rho_{\infty}$, and similarly for $\sigma$. We deduce that for a fixed $E\geq 0$,
		\begin{equation}\label{equation_liminf_machin}
			\liminf \left(c^{\mathrm{TF}}\int \rho_n^{5/3} + \int_{V\leq E} V \rho_n \right) \geq c^{\mathrm{TF}}\int \rho_{\infty}^{5/3} + \int_{V\leq E} V \rho_{\infty}.
		\end{equation}
		We set $I^E_n = \int_{V \geq E} (\rho_n+\sigma_n)$ and we have
		\begin{equation}
			E^{\text{TF}} + o(1) = \mathcal{E}_{(2)}^{\text{TF}}[\rho_n, \sigma_n] \geq \int_{V\geq E} V(\rho_n + \sigma_n) \geq E I^E_n.
		\end{equation} 
		Moreover, we have
		\begin{equation}
			E^{\text{TF}} + o(1) \geq \mathcal{E}_{(2)}^{\mathrm{TF}}[\rho_n\mathds{1}_{V\leq E}, \sigma_n \mathds{1}_{V\leq E}] + \int_{V\geq E}V(\rho_n + \sigma_n) \geq (1 - I_n^E)^{5/3} E^{\mathrm{TF}} + \int_{V\geq E}V(\rho_n + \sigma_n)
		\end{equation}
		and hence, for $E$ large but fixed, $(1 - I_n^E)^{5/3} \geq 1 - 2I_n^E$ which implies (by decomposing the integral)
		\begin{equation}
			E^{\mathrm{TF}} + o(1) \geq E^{\mathrm{TF}} + I_n^E(E/2 - 2E^{\mathrm{TF}}) + \frac{1}{2}\int_{V\geq E} V(\rho_n + \sigma_n).
		\end{equation}
		Then, if $E$ is large enough, $E/2 - 2E^{\mathrm{TF}} \geq 0$ and
		\begin{equation}\label{equation_convergence_integrale_depasse}
			\int_{V\geq E} V (\rho_n + \sigma_n) \to 0.
		\end{equation}
		Hence, for all $A \geq E$,
		\begin{equation}
			\int_{A\geq V\geq E} V(\rho_{\infty} + \sigma_{\infty}) = \lim_{n\to+\infty} \int_{A\geq V\geq E} V (\rho_n + \sigma_n) = 0,
		\end{equation}
		and thus
		\begin{equation}\label{equation_rho_infty_borne}
			\int_{V \geq E} V(\rho_{\infty} + \sigma_{\infty}) = 0.
		\end{equation}
		Therefore, combining \eqref{equation_liminf_machin} and \eqref{equation_rho_infty_borne}, we have 
		\begin{equation} E^{\text{TF}} + o(1) = \mathcal{E}^{\text{TF}}_{(2)}[\rho_n, \sigma_n] \geq \mathcal{E}^{\text{TF}}_{(2)}[\rho_{\infty}, \sigma_{\infty}] + o(1)
		\end{equation} 
		and thus 
		\begin{equation}\label{equation_rho_infty_TF}
			\rho_{\infty} = \sigma_{\infty} = \rho^{\text{TF}}/2.
		\end{equation} 
		As it is the only possible limit, the convergence holds without extracting. Moreover, by \eqref{equation_convergence_integrale_depasse}, \eqref{equation_rho_infty_borne}, \eqref{equation_rho_infty_TF} and weak-* convergence, we have
		\begin{equation} 
			\int V \rho_n \to \frac{1}{2}\int V \rho^{\mathrm{TF}}
		\end{equation}
		so 
		\begin{equation} 
			\int \rho_n^{5/3} \to \int \big(\rho^{\text{TF}}/2\big)^{5/3}, 
		\end{equation}
		which finally concludes the proof by Proposition 3.32 of \cite{brezis_functional_2011}, since the weak convergence and the convergence of the $L^{5/3}$ imply the strong convergence of $(\rho_n)$ and $(\sigma_n)$.\\~\\
		Step 2. If we take $\sigma_n = \frac{1}{2}\rho_n$, we have $\mathcal{E}^{\mathrm{TF}}_{(2)}[\sigma_n, \sigma_n] \to E^{\mathrm{TF}}$, and therefore 
		\begin{equation}
			\frac{1}{2}\rho_n = \sigma_n \to \frac{1}{2}\rho^{\mathrm{TF}}
		\end{equation}
		strongly in $L^{5/3}$ by Step 1.
	\end{proof}
	
	With our new notation, we can now express the perturbed two-spin Thomas-Fermi energy functional by
	\begin{equation}
		\mathcal{E}_N^{\mathrm{TF}}[\rup, \rdown] = \mathcal{E}^{\mathrm{TF}}_{(2)}[\rup, \rdown] + 8\pi a_wN^{1/3-\beta}\int \rup\rdown,
	\end{equation}
	and we finally prove that the perturbed Thomas-Fermi energy is equal to the Thomas-Fermi energy plus the interaction energy of the minimizer.
	
	\begin{proposition}[Equality between Thomas-Fermi energies] We recall that $E_N^{\mathrm{TF}}$ and $E^{\mathrm{TF}}$ are defined by \eqref{equation_definition_thomas_fermi_energy_deux_spin_perturbed} and \eqref{definition_energy_thomas_fermi_}. We have 
		\begin{equation}
			E^{\mathrm{TF}}_N = E^{\mathrm{TF}} + 2\pi a_w N^{1/3-\beta}\int \big(\rho^{\mathrm{TF}}\big)^2 + o(N^{1/3-\beta})
		\end{equation}
	\end{proposition}
	\begin{proof}
		First, we have 
		\begin{equation}
			E^{\mathrm{TF}}_N \leq \mathcal{E}^{\mathrm{TF}}_N\bigg[\frac{1}{2}\rho^{\mathrm{TF}}, \frac{1}{2}\rho^{\mathrm{TF}}\bigg] = E^{\mathrm{TF}} + 2\pi a_w N^{1/3-\beta}\int \big(\rho^{\mathrm{TF}}\big)^2.
		\end{equation}
		It remains to show the other inequality. Let $\big(\rho^{\mathrm{TF}}_N, \sigma^{\mathrm{TF}}_N\big)$ be a sequence of pair of densities such that  
		\begin{equation}
			\mathcal{E}^{\mathrm{TF}}_N[\rho^{\mathrm{TF}}_N, \sigma_N^{\mathrm{TF}}] = E^{\mathrm{TF}}_N + o(N^{4/3 - \beta})
		\end{equation}
		Then, it is clear that $(\rho^{\mathrm{TF}}_N, \sigma^{\mathrm{TF}}_N)$ is a minimizing sequence of $\mathcal{E}^{\mathrm{TF}}_{(2)}$. Because of Lemma \ref{lemme_convergence_suites_minimisantes}, this implies that $\rho^{\mathrm{TF}}_N$ and $\sigma^{\mathrm{TF}}_N$ converges strongly to $\rho^{\mathrm{TF}}/2$ in $L^{5/3}$, and hence there exists a subsequence of $\rho^{\mathrm{TF}}_N$ converging almost everywhere to $\rho^{\mathrm{TF}}$. Therefore, by Fatou's lemma,
		\begin{equation}
			\liminf\int \big(\rho^{\mathrm{TF}}_N\big)^2 \geq \int \big(\rho^{\mathrm{TF}}\big)^2,
		\end{equation}
		which implies that 
		\begin{equation}
			E^{\mathrm{TF}}_N \geq E^{\mathrm{TF}} + 2\pi a_w N^{1/3-\beta}\int \big(\rho^{\mathrm{TF}}\big)^2 + o(N^{1/3 - \beta}).
		\end{equation}
	\end{proof}
	
	\subsection{Scattering length}\label{subsection_diffusion_length}
	
	Here, we give some information on the scattering length of an interaction potential. In particular, we show that when the intensity of the potential tends to infinity, the potential behaves like a hard-core and its scattering length tends to its range. Then, we deduce the equivalent of Theorem \ref{theorem} when $w_N$ is more singular. We recall the definition of the scattering length $a_v$ of a radial positive potential of finite range $v$ on $\R^3$:
	\begin{equation}
		4\pi a_v = \inf\left\{\int |\nabla f|^2 + \frac{1}{2}v|f|^2,~ f\underset{\infty}{\to}1\right\}.
	\end{equation}
	Let $f_v$ be the minimizer of the variational problem, it is called the zero-energy scattering solution, and it satisfies
	\begin{equation}
		-\Delta f_v + \frac{1}{2}v f_v = 0.
	\end{equation}
	Moreover, we denote by $R_v$ the range of $v$ defined by
	\begin{equation}\label{equation_definition_range}
		R_v = \inf\big\{R\in \R_+,~ \supp v \subset B(0, R)\big\}.
	\end{equation}
	We recap some well-known properties of the zero-energy scattering solution in the following proposition (see for instance Theorem 5.2. of \cite{rougerie_scaling_2021} and Theorem A.1. of \cite{lieb_ground_2001}):
	
	\begin{proposition}[Properties of $f_v$]
		The zero-energy scattering solution is positive, radial and satisfies:
		\begin{equation}
			\begin{cases}
				f_v(x) = 1 - a_v/|x| &\mathrm{if}~|x|\geq R_v\\
				f_v(x) \geq 1 - a_v/|x| &\mathrm{else}.
			\end{cases} 
		\end{equation}
		Moreover, $f_v$ is bounded from above
		\begin{equation}
			f_v \leq 1.
		\end{equation}
	\end{proposition}
	Now, we can show that the scattering length of a potential whose intensity tends to infinity tends to the range of the interaction.
	\begin{lemma}[Scattering length of a potential of diverging intensity]\label{lemma_scattering_length_near_hardcore}
		\begin{equation}
			a_{Av} \to R_v~\text{when}~ A \to + \infty
		\end{equation}
	\end{lemma}
	\begin{proof}
		First, we show that in general, the scattering length of a potential is smaller thant its radius. Let
		\begin{equation}
			f_0(x) = \begin{cases}
				1 - R_v/|x|& \text{if}~ |x|\geq R_v\\
				0& \text{else}.
			\end{cases}
		\end{equation}
		Then, 
		\begin{equation}
			a_{Av} \leq \frac{1}{4\pi}\int \Big(|\nabla f_0|^2 + \frac{1}{2}A v|f_0|^2\Big) = \frac{1}{4\pi}\int_{x\geq R_v}|\nabla f_0|^2 = (R_v)^2\int_{R_v}^{+\infty}\frac{\d r}{r^2} = R_v.
		\end{equation}
		Thus, it remains to show that $a_{Av} \geq R_v + o(1)$. We fix $\varepsilon>0$ and assume that for a sequence $A_n \to + \infty$, we have 
		\begin{equation} 
			R_v - a_{A_nv} \geq \varepsilon. 
		\end{equation} 
		We denote by $f_n$ the minimizer of the scattering energy. Then, \begin{equation} 
			f_n(x)\geq 1 - a_{A_nv}/|x|. 
		\end{equation} 
		We set 
		\begin{equation} 
			C = \{x\in \R^3,~a_{A_nv}\leq |x| \leq R_v\}
		\end{equation} 
		and
		\begin{equation} 
			C_{\varepsilon} = \{x\in \R^3,~R_v - 2\varepsilon/3 \leq |x| \leq R_v - \varepsilon/3\}\subset C,
		\end{equation} 
		and we have
		\begin{equation}
			\int A_n v|f_n|^2\geq\int_{C} A_nv|f_n|^2 \geq A_n\int_{C_\varepsilon}v(x)\left(1 - \frac{a_{A_nv}}{R_v - 2\varepsilon/3}\right)^2\geq A_n \left(\frac{\varepsilon}{3R_v - 2\varepsilon}\right)^2\int_{C_\varepsilon}v \to + \infty.
		\end{equation}
		However, we know that $\int A_n v|f_n|^2 \leq 4\pi a_n \leq 4\pi R_v$, and therefore 
		\begin{equation}
			a_{A v} \underset{A\to + \infty}{\to}  R_v.
		\end{equation}
	\end{proof}
	We can now state the equivalent of Theorem \ref{theorem} for a more intense interaction potential.
	\begin{corollary}
		Let $1/3<\beta < 34/81$ and $w_N = N^{\alpha}w(N^{\beta}\cdot)$, with $\alpha > 2\beta-2/3$. Under (\ref{H_1}) and (\ref{H_2}), we have
		\begin{equation}
			E(N) = N E^{\mathrm{TF}} + 2\pi R_w N^{4/3-\beta}\int \big(\rho^{\mathrm{TF}}\big)^2 + o(N^{4/3-\beta}).
		\end{equation}
	\end{corollary}
	Indeed, let $\delta = \alpha - 2\beta-2/3$, then if we apply Theorem \ref{theorem} formally, we find (skipping the detail)
	\begin{align}
		E(N) &= N E^{\mathrm{TF}} + 2\pi a_{N^{\delta}w} N^{4/3-\beta}\int \big(\rho^{\mathrm{TF}}\big)^2 + o(N^{4/3-\beta}) \notag\\
		&= N E^{\mathrm{TF}} + 2\pi R_w(1 + o(1)) N^{4/3-\beta}\int \big(\rho^{\mathrm{TF}}\big)^2 + o(N^{4/3-\beta}).
	\end{align}
	A more rigorous derivation is possible by adapting the proof of Theorem \ref{theorem} and replacing $a_w$ by $R_w + o(1)$.
	
	\subsection{Some general notation}
	Here we recall some usual notation. Let $\Psi_{\Nup, \Ndown} \in \Hcal_{\Nup, \Ndown}$, we set $\Gamma_{\Nup, \Ndown} = |\Psi_{\Nup, \Ndown}\rangle\langle \Psi_{\Nup, \Ndown}|$ the density matrix associated with $\Psi_{\Nup, \Ndown}$. Then, we define the one-particle spin-up reduced density matrix $\gamma_{\Psi_{\Nup, \Ndown}}^{(1, 0)}$ by
	\begin{equation}
		\gamma_{\Psi_{\Nup, \Ndown}}^{(1, 0)}(x, y) = \Nup \int \overline{\Psi_{\Nup, \Ndown}(x, x_2,...,x_N)} \Psi_{\Nup, \Ndown}(y, x_2,..., x_N)\d x_2... \d x_N,
	\end{equation}
	acting on 
	\begin{equation}
		\Hcal = \left\{ \varphi \in H^1(\mathbb{R}^d),~ \int V |\varphi|^2 < +\infty\right\}.
	\end{equation}
	Similarly, we define the one-particle spin-down reduced density matrix $\gamma_{\Psi_{\Nup, \Ndown}}^{(0,1)}$. These one-particle reduced density matrices satisfy the Pauli principle 
	\begin{equation}
		0 \leq \gamma_{\Psi_{\Nup, \Ndown}}^{(1, 0)},\gamma_{\Psi_{\Nup, \Ndown}}^{(0,1)} \leq \mathds{1}.
	\end{equation}
	Finally, we set the one-particle spin-up reduced density by
	\begin{equation}
		\rho_{\Psi_{\Nup, \Ndown}}^{(1, 0)}(x) = \gamma_{\Psi_{\Nup, \Ndown}}^{(1, 0)}(x, x) = \Nup \int |\Psi_{\Nup, \Ndown}(x, x_2,...,x_N)|^2 \d x_2... \d x_N,
	\end{equation}
	which satisfies
	\begin{equation}
		\int \rho_{\Psi_{\Nup, \Ndown}}^{(1, 0)}(x) \d x = \Nup.
	\end{equation}

	\section{Upper bound}\label{section_upper_bound}
	In this section, we prove the following Proposition.
	
	\begin{proposition}[Upper bound on the energy]\label{proposition_upper_bound} Provided that $w\in C^{0}_c$, $V \in W^{1, \infty}_{\mathrm{loc}}$ and $1/3<\beta < 34/81$, we have the following inequality:
		\begin{equation}\label{equation_upper_bound}
			E(N) \leq N E^{TF} + 2\pi N^{4/3 - \beta}\int \big(\rho^{\mathrm{TF}}\big)^2 + o(N^{4/3 - \beta}).
		\end{equation}
	\end{proposition}
	
	To bound the ground state energy from above it is enough to find a state for which the energy is small enough. To do so, we decompose the space in small boxes so that the confining potential is almost constant in each box. Then, we create a state with particles of the lowest energy possible confined in the boxes. In each box, the number of particles is proportional to the local value of the Thomas-Fermi density $\rho^{\mathrm{TF}}$, and thus the state has a density which approximates $\rho^{\mathrm{TF}}$. Then we show that the energy of the particles in the small box can be mapped to the energy of another system, which can be bounded from above by \cite{lieb_ground-state_2005} (note that the bound we use relies on a similar decomposition). Nonetheless, instead of a thermodynamic limit, we use a combined limit; that is the density converges to 0 and the number of particles goes to infinity simultaneously. Therefore, we have to slightly adapt the inequality of \cite{lieb_ground-state_2005}.
	
	\subsection{Decomposition into boxes}
	
	\begin{figure}
		\centering
		\includegraphics[width=0.8\linewidth]{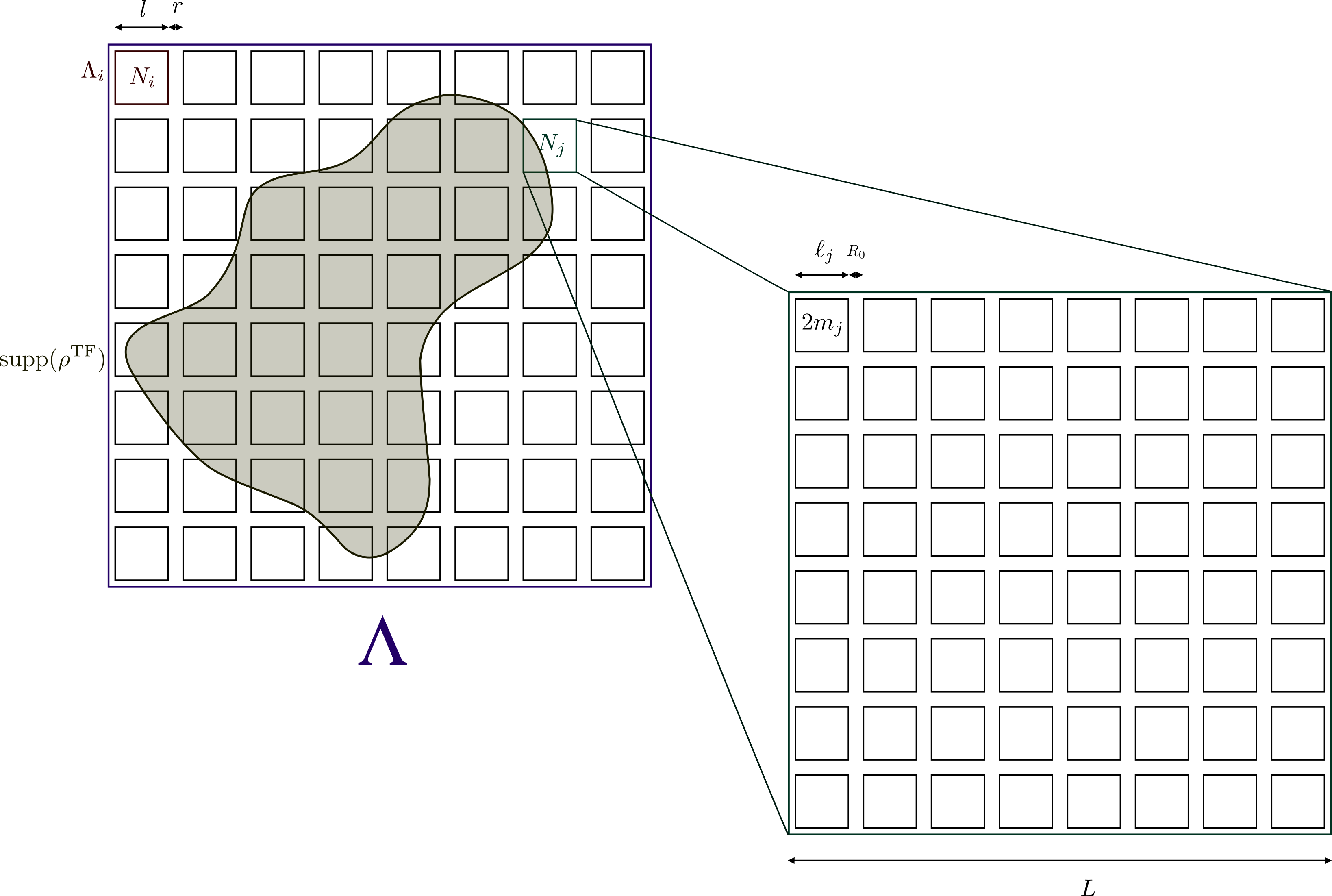}
		\caption{Representation of the double decomposition into boxes: for the second decomposition see the proof of Lemma \ref{lemma_energy_bound_small_box}. The grayed region is the support of $\rho^{\mathrm{TF}}$.}
	\end{figure}
	
	We have proved in Proposition \ref{proposition_properties_thomas_fermi_density} that $\rho^{\mathrm{TF}}$ is compactly supported, so we consider a cube $\Lambda\subset \R^3$ such that $\supp \rho^{\mathrm{TF}} \subset \Lambda$. For a given $l$ (which depends implicitly on $N$) and for $r= N^{-\beta}R_w$, we decompose $\Lambda$ into small boxes of size $l$ and at distance $r$ from one another, so that the potential is almost constant on each box and there is no interaction between two boxes. To be more precise, we have
	\begin{equation}
		\Lambda = \bigsqcup_{i\in I} \tilde{\Lambda}_i 
	\end{equation}
	with 
	\begin{equation}
		\tilde{\Lambda}_i = \prod_{k = 1}^3 \big[(\bar{x}_i)_k-(l+r)/2, (\bar{x}_i)_k + (l+r)/2 \big).
	\end{equation}	
	Moreover we set
	\begin{equation}
		\Lambda_i = \prod_{k = 1}^3 \big[(\bar{x}_i)_k-l/2, (\bar{x}_i)_k + l/2 \big).
	\end{equation}
	Let 
	\begin{equation}
		L = N^{\beta}l~~~~~\text{and}~~~~~\Lambda_L = [-L/2, L/2]^2.
	\end{equation}
	For the moment, we fix $N$ (and $\hbar$ accordingly). Let $N_i = 2M_i \in \N$, we consider
	\begin{equation}
		H^0_{N_i, \hbar} = \sum_{j = 1}^{N_i} \hbar^2 (-\Delta_j) + N^{2\beta - 2/3}\sum_{1 \leq j<k\leq N_i}w\big(N^{\beta}(x_j - x_k)\big) = H_{N_i, \hbar} - \sum_{j= 1}^{N_i} V(x_j)
	\end{equation}	
	and
	\begin{equation}
		\tilde{H}^0_{N_i, \hbar} = \sum_{j = 1}^{N_i}  (-\Delta_j) + \sum_{1 \leq j<k\leq N_i} w(x_j - x_k).
	\end{equation}	
	We want to prove the following proposition:
	
	\begin{proposition}[A change of variables]
		We have equality of the ground state energies:
		\begin{equation}
			\inf \Big\{ \langle \Psi | H_{N_i, \hbar}^0 | \Psi\rangle,~ \Psi \in L^2_{\mathrm{as}}\big(\Lambda_i^{M_i}\big)^{\otimes 2},~ \langle \Psi | \Psi\rangle = 1\Big\} = N^{2\beta - 2/3}\inf \Big\{\langle \Phi | \tilde{H}_{N_i, \hbar}^0 | \Phi\rangle,~ \Phi \in L^2_{\mathrm{as}}\big(\Lambda_L^{M_i}\big)^{\otimes2}, ~ \langle \Phi | \Phi\rangle = 1\Big\}.
		\end{equation}
	\end{proposition}
	\begin{proof}
		As $H^0_{N_i}$ is invariant by translation, we assume for simplicity that $\bar{x}_i = 0$. Let $\Psi \in L^2_{\mathrm{as}}\big(\Lambda_i^{M_i}\big)^{\otimes 2}\cap H^1\big(\Lambda_i^{N_i}\big)$ and $\Phi = \Psi(N^{-\beta}\cdot)$. First, it is clear that 
		\begin{equation}
			\langle \Phi|\Phi\rangle = N^{3N_i \beta}\langle \Psi|\Psi\rangle.
		\end{equation}
		Then, with a change of variables $y = N^{\beta}x$, we have
		\begin{align}
			\langle \Phi | \tilde{H}^0_{N_i, \hbar} |\Phi\rangle &= \int_{(\Lambda_L)^{N_i}} \Bigg(\sum_{j= 1}^{N_i} \big|\nabla_j\Phi(y)\big|^2 + \sum_{1\leq j < k \leq N} w(y_j - y_k)|\Phi(y)|^2\Bigg)\d y_1 ...\d y_{N_i}\notag\\
			&= \int_{(\Lambda_L)^{N_i}} \Bigg(\sum_{j= 1}^{N_i}N^{-2\beta} \big|\nabla_j\Psi(N^{-\beta}y)\big|^2 + \sum_{1\leq j < k \leq N} w(y_j - y_k)|\Psi(N^{-\beta}y)|^2\Bigg)\d y_1 ...\d y_{N_i}\notag\\
			&= N^{3N_i \beta} \int_{(\Lambda_i)^{N_i}} \Bigg(\sum_{j=1}^{N_i} N^{-2\beta} \big|\nabla_j \Psi(x)\big|^2 + \sum_{1\leq j < k \leq N} w\big(N^{\beta}(x_j - x_k)\big)|\Psi(x)|^2\bigg)\d x_1 ...\d x_{N_i}\notag\\
			&= N^{3N_i \beta} N^{2/3 - 2 \beta}\int_{(\Lambda_i)^{N_i}} \Bigg(\sum_{j=1}^{N_i} \hbar^2 \big|\nabla_j \Psi(x)\big|^2 + \sum_{1\leq j < k \leq N} N^{2\beta - 2/3} w\big(N^{\beta}(x_j - x_k)\big)|\Psi(x)|^2\bigg)\d x_1 ...\d x_{N_i}\notag\\
			&= \frac{\langle \Phi|\Phi\rangle}{\langle \Psi|\Psi\rangle} N^{2/3-2\beta} \langle \Psi|H^0_{N_i, \hbar}| \Psi\rangle.
		\end{align}	
		Therefore,
		\begin{equation}
			\frac{\langle \Phi|\tilde{H}^0_{N_i, \hbar}|\Phi\rangle}{\langle \Phi|\Phi\rangle} = N^{2/3 - 2\beta}\frac{\langle \Psi|H^0_{N_i, \hbar}|\Psi\rangle}{\langle \Psi |\Psi\rangle},
		\end{equation}
		which concludes the proof.
	\end{proof}	
	For the rest of the section, we define $M_i\in \mathbb{N}$ by
	\begin{equation}
		M_i = \left\lceil \frac{N}{2}\int_{\tilde{\Lambda}_i} \rho^{\mathrm{TF}}  \right\rceil,
	\end{equation}
	$\lceil\cdot\rceil$ being the ceiling function.
	
	\subsection{Upper bound in the boxes}

	In this section, we bound 
	\begin{equation} 
		\tilde{E}^0_{\hbar}(N_i, L) :=  \inf \big\{\langle \Phi | \tilde{H}_{N_i, \hbar}^0 | \Phi\rangle,~ \Phi \in L^2_{\mathrm{as}}\big(\Lambda_L^{M_i}\big)^{\otimes 2}, ~ \langle \Phi | \Phi\rangle = 1\big\}
	\end{equation}
	from above.
	
	\begin{lemma}[Energy bound in a small box]\label{lemma_energy_bound_small_box}
		Let $\varrho_i = M_i/L^3 > 0$ (i.e. $M_i \geq 1$). We have
		\begin{equation}\label{equation_energy_bound_small_box}
			\frac{1}{L^3}\tilde{E}^0_{\hbar}(N_i, L) \leq \big(2c^{\mathrm{TF}} \varrho_i^{5/3} + 8\pi a_w \varrho_i^2\big)\big(1 + o(\varrho_i^{1/3}) + O(\varrho_i^{-20/27}/N^{\beta}l)\big)
		\end{equation}
	\end{lemma}
	\begin{remark}
		For this lemma to be used, we need to check that $\varrho_i \to 0$ when $N \to + \infty$. By definition,
		\begin{equation}
			\varrho_i \leq N^{1-3\beta} \left(\frac{1}{|\Lambda_i|}\int_{\Lambda_i}\rho^{\mathrm{TF}} +1\right)\leq N^{1/3 - \beta}\Big(2\|\rho^{\mathrm{TF}}\|_{L^{\infty}} +1\Big).
		\end{equation}
		Moreover, we remark that $\varrho_i \to 0$ when $N\to + \infty$ uniformly in $i$, therefore (\ref{equation_energy_bound_small_box}) may be written
		\begin{equation}
			\frac{1}{L^3}\tilde{E}^0_{\hbar}(N_i, L) \leq 2c^{\mathrm{TF}} \varrho_i^{5/3} + 8\pi a_w \varrho_i^2\big(1 + o(1)\big),
		\end{equation}
		the $o(1)$ being uniform in $i$.
		\hfill $\diamond$
	\end{remark}
	
	\begin{proof}[Proof of Lemma \ref{lemma_energy_bound_small_box}.]
		Following \cite{lieb_ground-state_2005}, we decompose once again $\Lambda_L$ into small boxes of size $\ell_i$ and at distance $R_0$ from each other. We set $m_i = \big\lceil\varrho_i(\ell_i + R_0)^3 \big\rceil$.~\\~\\		
		We have the following equality which is a generalization of equation (10) of \cite{lieb_ground-state_2005}
		
		\begin{equation}
			\frac{1}{L^3}\tilde{E}^0_{\hbar}(N_i, L) \leq \frac{1}{(\ell_i + R_0)^3}\tilde{E}^0_{\hbar}(2 m_i, \ell_i)\big(1 + O(\ell_i/L)\big) \leq \frac{1}{\ell_i^3}\tilde{E}^0_{\hbar}(2 m_i, \ell_i)\big(1 + O(\ell_i/L)\big).
		\end{equation}	
		Then, by taking 
		\begin{equation}
			\ell_i \propto \varrho_i^{-20/27},
		\end{equation}		
		we have, following \cite{lieb_ground-state_2005},
		
		\begin{equation}
			\frac{1}{\ell^3}\tilde{E}^0_{\hbar}(m_i, \ell) \leq 2c^{\mathrm{TF}} \varrho_i^{5/3} + 8\pi a_w \varrho_i^2 + o(\varrho_i^2),
		\end{equation}		
		which gives (\ref{equation_energy_bound_small_box}).		
	\end{proof}
	
	\begin{remark}
		Here, we can easily use the result of \cite{lieb_ground-state_2005} because the bound on $\tilde{E}^0_{\hbar}(N_i, L)$ is obtained via the energy of $m_i$ particles, and $m_i$ only depends on $\varrho_i$.
		\hfill $\diamond$
	\end{remark}
	\subsection{Conclusion}
	
	There remains to construct a state with at least $N$ particles (as per \eqref{equation_energie_croit_avec_nombre} below we do not need to have exactly $N$ particles in the trial state) whose energy is correctly bounded from above. Let $N' = \sum_{i\in I}N_i$. We have
	\begin{equation}
		N' \geq N \sum_{i\in I} \int_{\tilde{\Lambda}_i}\rho^{\mathrm{TF}} = N\int_{\Lambda} \rho^{\mathrm{TF}} = N.
	\end{equation}
	For each $i\in I$, let $\Psi_i\in L^2_{\mathrm{as}}\big(\Lambda_i^{M_i}\big)^{\otimes 2}$ minimizing $H_{N_i}^0$ on $\Lambda_i$ with Dirichlet boundary conditions of norm 1 More precisely, we take $\Psi_i$ such that 
	\begin{equation}
		\langle \Psi_i | H_{N_i,\hbar}^0|\Psi_i\rangle \leq N^{2\beta - 2/3}\tilde{E}_{\hbar}^0(N_i, L) + C\frac{\varrho_i^{25/27}}{L}.
	\end{equation}
	Then, we set 
	\begin{equation}
		\Psi_{N'} = \bigwedge_{i\in I} \Psi_i.
	\end{equation}	
	As the $\Psi_i$ are localized in the $\Lambda_i$, we have
	\begin{equation}
		\langle \Psi_{N'}|H_{N',\hbar}^0|\Psi_{N'}\rangle = \sum_{i\in I} \langle \Psi_i| H_{N_i, \hbar}^0|\Psi_i\rangle.
	\end{equation}	
	It is clear that for a given $\hbar$, the energy increases with the number of particles; we have
	\begin{align}\label{equation_energie_croit_avec_nombre}
		E(N) &\leq \langle \Psi_{N'}| H_{N, \hbar} | \Psi_{N'}\rangle \notag \\ &= \int\bigg(\sum_{j=1}^N N^{-2/3}|\nabla_j \Psi_{N'}(X_{N'})|^2 + \sum_{j=1}^N V(x_j)|\Psi_{N'}(X_{N'})|^2 + \sum_{1\leq j < k \leq N} w_N(x_j - x_k)|\Psi_{N'}(X_{N'})\bigg)\d X_{N'} \notag\\
		&\leq \int\bigg(\sum_{j=1}^{N'} N^{-2/3}|\nabla_j \Psi_{N'}(X_{N'})|^2 + \sum_{j=1}^{N'} V(x_j)|\Psi_{N'}(X_{N'})|^2 + \sum_{1\leq j < k \leq N'} w_N(x_j - x_k)|\Psi_{N'}(X_{N'})\bigg)\d X_{N'} \notag\\
		&= \langle \Psi_{N'}|H_{N', \hbar}|\Psi_{N'}\rangle.
	\end{align}	
	There only remains to bound the energy of $\Psi_{N'}$. We will treat separately the potential energy and the kinetic and interaction energies.~\\~\\
		$\bullet$ On the one hand, 
		\begin{equation}
			\langle \Psi_{N'}|H_{N', \hbar}^0|\Psi_{N'}\rangle \leq \sum_{i\in I} \langle \Psi_i| H_{N_i, \hbar}^0|\Psi_i\rangle \leq N^{2\beta - 2/3}\sum_{i\in I}\left(\tilde{E}^0_{\hbar}(N_i, L) + C\frac{\varrho_i^{25/27}}{L}\right). 
		\end{equation}
		By using Lemma \ref{lemma_energy_bound_small_box}, we get
		\begin{equation}
			\langle \Psi_{N'}|H_{N', \hbar}^0|\Psi_{N'}\rangle \leq N^{2\beta - 2/3}L^3\sum_{i\in I}\left(2 c^{\mathrm{TF}} \varrho_i^{5/3} + 8\pi a_w\varrho_i^2 + o(\varrho_i^2) + C\frac{\varrho_i^{25/27}}{L} \right).
		\end{equation}	
		Recall that $\varrho_i = M_i/L^3$ to find
		\begin{equation}
			\langle \Psi_{N'}|H_{N', \hbar}^0|\Psi_{N'}\rangle \leq N^{2\beta - 2/3} \sum_{i\in I} \left(2 c^{\mathrm{TF}} L^{-2}M_i^{5/3} + 8\pi a_wL^{-3}M_i^2\big(1 + o(1)\big) + C L^{-21/27}M_i^{25/27}\right).
		\end{equation}	
		For all $i\in I$, $M_i \leq \frac{N}{2}\int_{\tilde{\Lambda_i}}\rho^{\mathrm{TF}} + 1$, hence for $p\in \{5/3, 2\}$,	
		\begin{equation}
			M_i^p \leq \left(\frac{N}{2}\int_{\tilde{\Lambda}_i}\Big(\rho^{\mathrm{TF}}+ 2N^{-1}(l+r)^{-3}\Big)\right)^p\leq 2^{-p}N^p (l+r)^{3p - 3}\int_{\tilde{\Lambda_i}}\Big(\rho^{\mathrm{TF}} + 2N^{-1}(l + r)^{-3}\Big)^p
		\end{equation}
		by Jensen inequality. Moreover, 
		\begin{equation}
			M_i^{25/27}\leq 1 + \frac{25}{27}\frac{N}{2}\int_{\tilde{\Lambda}_i}\rho^{\mathrm{TF}}
		\end{equation}
		Therefore, as $L = N^{\beta}l$,
		\begin{align}
			\langle \Psi_{N'}|H_{N', \hbar}^0|\Psi_{N'}\rangle &\leq 2^{-2/3}c^{\mathrm{TF}} N \big(1 + r/l\big)^2 \int_{\Lambda}\Big(\rho^{\mathrm{TF}} + 2N^{-1}(l + r)^{-3}\Big)^{5/3} \notag \\
			&+ 2\pi R_wN^{4/3 - \beta}\big(1 + r/l\big)^3\big(1 + o(1)\big)\int_{\Lambda}\Big(\rho^{\mathrm{TF}} + 2N^{-1}(l + r)^{-3}\Big)^2\notag\\
			&+ C l^{-21/27}N^{-21\beta/27}\big(|I| + N \|\rho^{\mathrm{TF}}\|_{L^1}\big).
		\end{align}
		For $p\in \{5/3, 2\}$,
		\begin{equation}
			\big\|2N^{-1}(l+r)^{-3}\big\|_{L^p(\Lambda)} = 2|\Lambda|^{1/p}N^{-1}(l+r)^{-3} = o(N^{1/3 - \beta})
		\end{equation}
		as long as $l \gg N^{\beta/3 - 4/9}$. Moreover, since $l \gg \hbar$, we have $r/l = o(N^{1/3 - \beta})$. This implies that
		\begin{equation}
			\langle \Psi_{N'}|H_{N', \hbar}^0|\Psi_{N'}\rangle \leq 2^{-2/3}c^{\mathrm{TF}} N\int_{\Lambda}\big(\rho^{\mathrm{TF}}\big)^{5/3} \\
			+ 2\pi R_wN^{4/3 - \beta}\int_{\Lambda}\big(\rho^{\mathrm{TF}}\big)^2 + o(N^{4/3 - \beta})
		\end{equation}
		as long as 
		\begin{equation}l^{-21/27}N^{-21\beta/27}\big(|I| + N \|\rho^{\mathrm{TF}}\|_{L^1}\big) = o(N^{4/3-\beta}).\end{equation} 
		If $l^{-3} = o(N)$, this corresponds to 
		\begin{equation}l^{21/27}\gg N^{-21\beta/27 - 1 + 3\beta} ~~~\text{i.e.}~~~ l\gg N^{-\beta +27(3\beta-1)/21}.
		\end{equation}	
		$\bullet$ On the other hand,
		\begin{equation}
			\bigg\langle\Psi_{N'}\bigg|\sum_{j=1}^{N'} V(x_j)\bigg|\Psi_{N'}\bigg\rangle \leq \sum_{i\in I} N_i \sup_{\Lambda_i}V\leq  \sum_{i\in I}\Big(N  \sup_{\tilde{\Lambda}_i}V\int_{\tilde{\Lambda}_i}\rho^{\mathrm{TF}} + 2\Big).
		\end{equation}
		For $x\in \tilde{\Lambda}_i$, we have
		\begin{equation}
			\sup_{\tilde{\Lambda_i}}V \leq V(x) - \|\nabla V\|_{L^{\infty}(\Lambda)}l,
		\end{equation}
		and therefore
		\begin{equation}
			\bigg\langle\Psi_{N'}\bigg|\sum_{j=1}^{N'} V(x_j)\bigg|\Psi_{N'}\bigg\rangle \leq N\int_{\Lambda} V\rho^{\mathrm{TF}} - N \|\nabla V\|_{L^{\infty}(\Lambda)}l \int_{\Lambda}\rho^{\mathrm{TF}} + o(N^{4/3 - \beta}),
		\end{equation}
		provided that $|I| = O(l^{-3}) = o(N^{4/3 - \beta})$. As long as $l \ll N^{1/3 - \beta}$, we finally find
		\begin{equation}
			\bigg\langle\Psi_{N'}\bigg|\sum_{j=1}^{N'} V(x_j)\bigg|\Psi_{N'}\bigg\rangle \leq N\int_{\Lambda} V\rho^{\mathrm{TF}} + o(N^{4/3 - \beta}).
		\end{equation}
		$\bullet$ We have proved that, as long as $N^{1/3 - \beta} \gg l \gg N^{-\frac{27}{21}(1-3\beta) - \beta},~ N^{\beta/3 - 4/9}$, (\ref{equation_upper_bound}) holds. Therefore, it only remains to check that for $\beta < 34/81$, we have 
		\begin{equation} \label{equation_echelle_possible_l}
			N^{1/3 - \beta} \gg N^{-\frac{27}{21}(1-3\beta) - \beta},~ N^{\beta/3 - 4/9}.
		\end{equation}
		Indeed,
		\begin{equation}
			N^{1/3 - \beta} \gg N^{-\frac{27}{21}(1-3\beta) - \beta} \iff \frac{1}{3} - \beta > -\frac{27}{21}(1-3\beta) - \beta \iff \beta < \frac{34}{81}\approx 0.42
		\end{equation}
		and 
		\begin{equation}
			N^{1/3 - \beta} \gg N^{\beta/3 - 4/9} \iff \frac{1}{3} - \beta > \frac{\beta}{3} - \frac{4}{9} \iff \beta < \frac{7}{12}.
		\end{equation}
		This means that as long as $\beta < 34/81$, \eqref{equation_echelle_possible_l} is satisfied and we can find $l$ such that 
		\begin{equation}
			N^{1/3 - \beta} \gg l \gg N^{-\frac{27}{21}(1-3\beta) - \beta},~ N^{\beta/3 - 4/9}.
		\end{equation}
	This concludes the proof of Proposition \ref{proposition_upper_bound}.
	
	\section{Lower bound}\label{section_lower_bound}
	
	In this section, we prove the following proposition.
	\begin{proposition}[Lower bound on the energy]\label{proposition_lower_bound}
		For $1/3 < \beta <1/2$, and assuming \emph{(}\ref{H_1}\emph{)} and \emph{(}\ref{H_2}\emph{)}, we have the following inequality:
		\begin{equation}
			E(N)\geq N E^{\mathrm{TF}} + 2\pi a_w N^{4/3-\beta}\int \big(\rho^{\mathrm{TF}}\big)^2 + o(N^{4/3-\beta}).
		\end{equation}
	\end{proposition}
	
	To prove Proposition \ref{proposition_lower_bound}, we have to check that we have a lower bound on the energy for all states, or for all states with low enough energy. For all this section, for each $N$, we fix a $\Psi_{\Nup, \Ndown} \in \Hcal_{\Nup, \Ndown}$ that almost minimizes $H_N$. To be more precise, we impose
	\begin{equation}\label{equation_supposition_etat_test}
		\langle\Psi_{\Nup, \Ndown}|H_N |\Psi_{\Nup, \Ndown}\rangle = E(N) + o(N^{4/3 - \beta}).
	\end{equation}
	In particular, because of Proposition \ref{proposition_upper_bound}, we have
	\begin{equation}\label{equation_bound_energy_test}
		\langle\Psi_{\Nup, \Ndown}|H_N |\Psi_{\Nup, \Ndown}\rangle \leq N E^{\mathrm{TF}} + O(N^{4/3-\beta}).
	\end{equation}
	Since the interaction between particles of same spin does not appear in the limit, we remove these interactions from the model. More precisely, we have 
	\begin{equation}
		H_N \geq \sum_{j=1}^N \Big(-\hbar^2\Delta_j + V(x_j)\Big) + \sum_{j=1}^{\Nup}\sum_{k=N_{\uparrow} + 1}^Nw_N(x_j - x_k).
	\end{equation}
	To recover the Thomas-Fermi energy, we divide the Hamiltonian into two parts that give the unperturbed Thomas-Fermi energy and the interaction energy respectively. Let
	\begin{equation}
		\Gamma(p) = \max\left\{1 - \frac{p_F^2}{|p|^2}, 0\right\}
	\end{equation}
	for a given $p_F\in \mathbb{R}_+$ that will depend on $N$. We decompose the kinetic energy in a low-frequency and a high-frequency parts.
	\begin{equation}
		H_N \geq H_N^{\mathrm{LF}} + H_{\Nup, \Ndown}^{\uparrow} + H_{\Nup, \Ndown}^{\downarrow},
	\end{equation}
	with
	\begin{equation}
		H_N^{\mathrm{LF}} = \sum_{j = 1}^N -\hbar^2 \nabla_j \big(1 - \Gamma(p_j)\big) \nabla_j + \sum_{j=1}^N V(x_j)
	\end{equation}
	and
	\begin{equation}
		H_{\Nup, \Ndown}^{\uparrow} = \sum_{j = 1}^N -\hbar^2 \nabla_j\Gamma(p_j)\nabla_j + \frac{1}{2}\sum_{j=1}^{\Nup}\sum_{k= N_{\uparrow} + 1}^N w_N(x_j - x_k).
	\end{equation}
	Let us now give a quick outline of the proof of the lower bound. First, in Section \ref{subsection_low_frequency}, we give a lower bound for $H_N^{\mathrm{LF}}$. More specifically, we show that it is bounded from below by $N E^{\mathrm{TF}}$ and an error term. Then, we have to bound $\langle \Psi_{\Nup, \Ndown}|H_{\Nup, \Ndown}^{\uparrow} \Psi_{\Nup, \Ndown}\rangle$ from below. To do so, we use the Dyson lemma in Section \ref{subsection_dyson} in order to replace $w_N$ by a less singular potential. In particular, the Dyson lemma makes $a_w$ appear naturally. To use this lemma, we have to use the higher frequencies of the kinetic energy, which is why we removed it from $H_N^{\mathrm{LF}}$. After using the Dyson lemma, we rewrite what we obtain, and have to deal with several error terms in Section \ref{subsection_error_estimates}. Then, in Section \ref{subsection_convergence_densities}, we estimate the main term for the interaction, and in particular, we show the convergence of the density of $\Psi_{\Nup, \Ndown}$ and of the projector on the $\Nup$ first eigenvalues of the free Hamiltonian towards the Thomas-Fermi density. We can then conclude in Section \ref{subsection_conclusion}, and explain the constraints on $\beta$.

	\subsection{Main term - low frequencies and potential energy}\label{subsection_low_frequency}
	
	We recall that
	\begin{equation}
		H_N^{\mathrm{LF}} = \sum_{j = 1}^N -\hbar^2 \nabla_j \big(1 - \Gamma(p_j)\big) \nabla_j + \sum_{j=1}^N V(x_j) =: \sum_{j=1}^N (H_0^{\mathrm{LF}})_j.
	\end{equation}
	
	\begin{lemma}[Main term of the energy]\label{lemma_low_frequency}
		For any sequence of state $(\Phi_{\Nup, \Ndown})$ such that $\tr(-\hbar^2 \Delta\gamma_{\Phi_{\Nup, \Ndown}}^{(1)}) = O(N)$, we have the following inequality
		\begin{equation}
			\langle \Phi_{\Nup, \Ndown}| H_N^{\mathrm{LF}}|\Phi_{\Nup, \Ndown}\rangle \geq N E^{\mathrm{TF}} + O(N p_F^{-2}) + O(N^{5/6}).
		\end{equation}
	\end{lemma}
	
	\begin{proof} 
		For such a sequence $(\Phi_{\Nup, \Ndown})$, we denote by $\gamma_{\Phi_{\Nup, \Ndown}}^{(1, 0)}$ and $\gamma_{\Phi_{\Nup, \Ndown}}^{(1, 0)}$ the one particle reduced density matrixes associated. Let 
		\begin{equation} 
			m_{\uparrow}(x, p) = \langle f_{x,p}^{\hbar} | \gamma_{\Phi_{\Nup, \Ndown}}^{(1, 0)} | f_{x,p}^{\hbar} \rangle ~~~\text{and}~~~ m(x, p) = m_{\uparrow}(x, p) + m_{\downarrow}(x, p),
		\end{equation}
		then we have (see (\ref{equation_potential_energy_husimi}) and Lemma \ref{lemma_kinetic_energy_low_frequency_husimi})
		\begin{equation}
			\langle \Phi_N | H_N^{\mathrm{LF}} |\Phi_N\rangle = \tr(\gamma_{\Nup, \Ndown}^{(1)} H_0^{\mathrm{LF}}) \geq \frac{N}{(2\pi)^3}\iint\big(|p|^2 (1 - \Gamma(p)\big) + V(x)\big) m(x, p)\d x \d p + O(N^{5/6}). 
		\end{equation}
		We set 
		\begin{equation} 
			\tilde{m}_{\uparrow\downarrow}(x, p) = \mathds{1}_{|p| \leq C \rho_{m_{\uparrow\downarrow}}(x)^{1/3}}~~~ \text{where}~~~ \rho_{m_{\uparrow\downarrow}}(x) = \int m_{\uparrow\downarrow}(x, p)\d p~~~\text{and}~~~ C= (3/4\pi)^{1/3}
		\end{equation}
		so that $\int \tilde{m}_{\uparrow\downarrow}(x, p)\d p = \rho_{m_{\uparrow\downarrow}}(x)$. As $p\in\R_ + \mapsto p^2(1 - \Gamma(p))$ is an increasing function, we get
		\begin{equation}
			\frac{1}{N}\langle \Phi_N | H_N^{\mathrm{LF}}| \Phi_N\rangle \geq \int V\rho + \iint |p|^2 (1 - \Gamma(p))\tilde{m}(x, p) \d p \d x + O(N^{5/6}).
		\end{equation}
		Let $p_F\in \R_+$ be the Fermi impulsion, that we will chose at the end of the proof in \eqref{equation_choix_constantes}. By separating the integral depending on wether $|p| \geq p_F$ or not, and by injecting the expression of $\Gamma$, we finally find
		
		\begin{multline}
			\frac{1}{N}\langle \Phi_N | H_N^{\mathrm{LF}} \Phi_N\rangle \geq \int V\big(\rho_{m_{\uparrow}} + \rho_{m_{\downarrow}} \big)+  \int \Big(c^{TF}\rho_{m_{\uparrow}}^{5/3}\mathds{1}_{C \rho_{m_{\uparrow}}^{1/3}\leq p_F} + c^{TF}\rho_{m_{\downarrow}}^{5/3}\mathds{1}_{C \rho_{m_{\downarrow}}^{1/3}\leq p_F} \\
			+ \big(p_F^2 \rho_{m_{\uparrow}} -\frac{8\pi}{15}p_F^5\big)\mathds{1}_{C \rho_{m_{\uparrow}}^{1/3} > p_F} + \big(p_F^2 \rho_{m_{\downarrow}} -\frac{8\pi}{15}p_F^5\big)\mathds{1}_{C \rho_{m_{\downarrow}}^{1/3} > p_F}\Big) + O(N^{5/6}) =: \mathcal{E}^{\mathrm{TF}}_{p_F}[\rho_{m_{\uparrow}},\rho_{m_{\downarrow}}] + O(N^{5/6}).
		\end{multline}
		We set
		\begin{equation}
			\mathcal{E}^{\mathrm{TF}}_{p_F}[\rho] = \int V\rho +  \int \Big(c^{TF}\rho^{5/3}\mathds{1}_{C \rho^{1/3}\leq p_F} + \big(p_F^2 \rho -\frac{8\pi}{15}p_F^5\big)\mathds{1}_{C \rho^{1/3} > p_F}\Big), 
		\end{equation}
		and
		\begin{equation}
			E^{\mathrm{TF}}_{p_F} = \inf\left\{\mathcal{E}^{\mathrm{TF}}_{p_F}[\rho],~ \rho\geq 0,~\int \rho = 1\right\} = \inf\left\{\mathcal{E}^{\mathrm{TF}}_{p_F}[\rup,\rdown],~ \rup,\rdown\geq 0,~\int \big(\rup+\rdown\big) = 1\right\}
		\end{equation}	
		Then, we have
		\begin{equation}
			\langle \Phi_{\Nup, \Ndown} | H_N^{\mathrm{LF}}| \Phi_{\Nup, \Ndown}\rangle \geq 
			NE^{\mathrm{TF}}_{p_F} + O(N^{5/6}).
		\end{equation}	
		Note that there exists a unique minimizer of $\mathcal{E}^{\mathrm{TF}}_{p_F}$, that we denote $\Bar{\rho}$. As $E^{\mathrm{TF}}_{p_F} \leq E^{\mathrm{TF}}$, $\Bar{\rho}$ satisfies
		\begin{equation}
			\frac{4\pi}{5}\int \mathds{1}_{\Bar{\rho}^{1/3} \geq p_F} \leq \frac{E^{\mathrm{TF}}}{p_F^5}.
		\end{equation}		
		Then, we also have
		\begin{equation}
			\int \bar{\rho} \mathds{1}_{\Bar{\rho}^{1/3} \geq p_F} \leq \frac{E^{\mathrm{TF}}}{p_F^2}+ \frac{8\pi p_F^3}{15} \frac{5}{4\pi}\frac{E^{\mathrm{TF}}}{p_F^5} = \frac{5}{3}\frac{E^{\mathrm{TF}}}{p_F^2}.
		\end{equation}	
		Let $I = \int \Bar{\rho}\mathds{1}_{\Bar{\rho}^{1/3} \leq p_F}$ and $\tilde{\rho} = \Bar{\rho}\mathds{1}_{\Bar{\rho}^{1/3} \leq p_F}$, we have 
		\begin{equation}
			\int \tilde{\rho} = I \geq 1 - \frac{5 E^{\mathrm{TF}}}{3 p_F^2}~~~\mathrm{and}~~~\mathcal{E}^{\mathrm{TF}}[\tilde{\rho}] = \mathcal{E}_{p_F}^{\mathrm{TF}}[\tilde{\rho}] \leq E^{\mathrm{TF}}_{p_F}.
		\end{equation}		
		We set $\hat{\rho} = I^{-1}\tilde{\rho}$ and we finally find		
		\begin{equation}
			E^{\mathrm{TF}} \leq \mathcal{E}^{\mathrm{TF}}[\hat{\rho}] \leq I^{-5/3}\mathcal{E}^{\mathrm{TF}}[\tilde{\rho}] \leq I^{-5/3} E_{p_F}^{\mathrm{TF}} \leq E_{p_F}^{\mathrm{TF}}\big(1 + O(p_F^{-2})\big).
		\end{equation}		
		Therefore,
		\begin{equation}
			\langle \Phi_{\Nup, \Ndown}| H_N^{\mathrm{LF}}\Phi_{\Nup, \Ndown}\rangle \geq N E^{\mathrm{TF}} + O(N p_F^{-2}) + O(N^{5/6}).
		\end{equation}
	\end{proof}
	Because of the a priori bound (\ref{equation_bound_energy_test}), $(\Psi_{\Nup, \Ndown})$ satisfies the hypotheses of Lemma \ref{lemma_low_frequency}, and therefore,
	\begin{equation}\label{equation_conclusion1}
		\langle \Psi_{\Nup, \Ndown}| H_N^{\mathrm{LF}}|\Psi_{\Nup, \Ndown}\rangle \geq N E^{\mathrm{TF}} + O(N p_F^{-2}) + O(N^{5/6}).
	\end{equation}
	Now, we have to bound $\langle \Psi_{\Nup, \Ndown}|H_{\Nup,\Ndown}^{\uparrow}|\Psi_{\Nup, \Ndown}\rangle$ and $\langle \Psi_{\Nup, \Ndown}|H_{\Nup,\Ndown}^{\downarrow}|\Psi_{\Nup, \Ndown}\rangle$ from below.
	
	\subsection{Dyson lemma: high frequencies and interactions}\label{subsection_dyson}
	The Dyson lemma is a well-known result. It allows us to make a potential less singular, at the cost of the high frequencies of the kinetic energy. Here, we will use it to soften the interaction potential $w_N$ and to exhibit the scattering length. Corollary 1 of \cite{lieb_ground-state_2005} can be written:
	\begin{lemma}[Dyson lemma]\label{lemma_Dyson_general}
		Let $R\geq R_0\geq R_v$. If $y_1,...,y_N \in \mathbb{R}^3$ satisfy $|y_i - y_j|\geq 2 R$ for all $i\neq j$, then
		\begin{equation}
			-\nabla \chi_s(p)^2 \nabla + \frac{1}{2}\sum_{i=1}^N v(x - y_j) \geq \sum_{i=1}^N \Big((1-\varepsilon)a_v U_R(x-y_i) - \frac{a(v)}{\varepsilon} u_R(x-y_i)\Big),
		\end{equation}
		with 
		\begin{equation}
			U_R(x) = \mathds{1}_{R_0 \leq |x| \leq R}\frac{3}{R^3 - R_0^3},
		\end{equation}
		\begin{equation}
			\chi_s(p) = l(sp) =  (sp - 1)\mathds{1}_{sp\geq 1} - (sp - 2) \mathds{1}_{sp \geq 2} = (sp-1)\mathds{1}_{1 \leq sp \leq 2} + \mathds{1}_{sp\geq2},
		\end{equation}
		\begin{equation}
			|u_R(x)|\leq C\frac{R^2}{s^5},~~~ \int |u_R(x)| \d x\leq C \frac{R^2}{s^2}~~~\mathrm{and}~~~\sum_{i=1}^N u_R(x - y_i)\leq \frac{C}{Rs^2}.
		\end{equation}
	\end{lemma}
	
	From this, we can deduce the following version of the Dyson lemma:
	
	\begin{lemma}[Dyson lemma in the semi-classical scaling]\label{lemma_dyson_our_version}
		Let $\hbar \geq R\geq R_0\geq N^{-\beta}R_w$. If $y_1,...y_N$ statisfy $|y_i-y_j| \geq 2R$ for all $i\neq j$, then
		\begin{equation}
			-\hbar^2 \nabla \chi_s(p)^2 \nabla + \frac{1}{2}\sum_{i=1}^N w_N(x - y_j) \geq N^{-\beta-2/3} \sum_{i=1}^N \Big((1-\varepsilon)4\pi a_w U_R(x-y_i) - \frac{a_w}{\varepsilon} u_R(x-y_i)\Big),
		\end{equation}
		with 
		\begin{equation}\label{equation_definition_UR}
			U_R(x) = \mathds{1}_{R_0 \leq |x| \leq R}\frac{3}{4\pi(R^3 - R_0^3)},
		\end{equation}
		\begin{equation}
			\chi_s(p) = l(sp) =  (sp - 1)\mathds{1}_{sp\geq 1} - (sp - 2) \mathds{1}_{sp \geq 2} = (sp-1)\mathds{1}_{1 \leq sp \leq 2} + \mathds{1}_{sp\geq2},
		\end{equation}
		\begin{equation}\label{equation_properties_uR}
			|u_R(x)|\leq C\frac{R^2}{s^5},~~~ \int |u_R(x)| \d x\leq C \frac{R^2}{s^2}~~~\mathrm{and}~~~\sum_{i=1}^N u_R(x - y_i)\leq \frac{C}{Rs^2}.
		\end{equation}
	\end{lemma}
	Note that in particular, $\int U_R = 1$.
	\begin{proof}
		To deduce Lemma \ref{lemma_dyson_our_version} from Lemma \ref{lemma_Dyson_general}, we only have to prove that $a(\hbar^{-2}w_N) = N^{-\beta}a_w$. Recall that, by definition,
		\begin{equation}
			4\pi a_{\hbar^{-2}w_N} = \hbar^{-2}\inf\left\{ \int \Big(\hbar^2 |\nabla f|^2+ \frac{1}{2}w_N|f|^2\Big),~ f\underset{\infty}{\to}1 \right\}.
		\end{equation}
		For $f$ satisfying $f(x) \to 1$ when $|x| \to + \infty$, we set $\tilde{f} = f(N^{-\beta}\cdot)$ and we have
		\begin{equation}
			\int \Big(\hbar^2 |\nabla f|^2+ \frac{1}{2}w_N|f|^2\Big) =  N^{-\beta-2/3}\int \Big( |\nabla \tilde{f}|^2+ \frac{1}{2}w|\tilde{f}|^2\Big).
		\end{equation}
		Then, we have
		\begin{equation}
			4\pi a_{\hbar^{-2}w_N} = N^{-\beta}\inf\left\{\int \Big(|\nabla\tilde{f}|^2 + \frac{1}{2}w|\tilde{f}|^2\Big),~\tilde{f} \underset{\infty}{\to} 1\right\} = 4\pi N^{-\beta} a_w.
		\end{equation}
	\end{proof}
	\begin{remark}~
		\begin{itemize}
			\item Note that if $w_N = N^{\alpha}w_N(N^{\beta}\cdot)$ with $\alpha > 2\beta-2/3$, we find $a_{\hbar^{-2}w_N} = N^{-\beta}R_w + o(1)$.
			\item In the following, we fix $R$ and $R_0$ such that $\hbar \gg R\geq R_0\gg N^{-\beta}R_w$.
		\end{itemize}	\hfill $\diamond$
	\end{remark}
	Now, we want to choose $s$ such that the kinetic energy required in the lemma corresponds to the part we did not use for $H_N^{\mathrm{LF}}$. Let $s\leq p_F^{-1}$, we have
	\begin{equation}
		\Gamma \geq (1 - s^2p_F^2)\chi_s ^2.
	\end{equation}
	\begin{definition}\label{definition_plein_de_trucs}
		Let $Y = (y_1,..., y_{\Ndown})\in \R^{3\Ndown}$ and $\tilde{Y} = \{y_j \in Y, ~\forall k\neq j,~|y_j - y_k|\geq 2R\}$ the set of positions $y_j$ at distance at least $2R$ of all the other.
		We set
		\begin{equation}\label{equation_definition_W}
			W_Y^+ = 4\pi a_w (1 - \varepsilon)\sum_{k,y_k\in\tilde{Y}}U_R(\cdot-y_k),~~~W_Y^- = \frac{a_w}{\varepsilon}\sum_{k,y_k\in\tilde{Y}}u_R(\cdot-y_k)~~~\text{and}~~~W_Y = W_Y^+ - W_Y^-.
		\end{equation}
		Moreover, let 
		\begin{equation}\label{equation_definition_nY}
			n_Y = \int |\Psi_{\Nup, \Ndown}(X, Y)|^2 \d X
		\end{equation}	
		and 
		\begin{equation}\label{equation_definition_gammaY}
			\gamma_Y(x, x') = \frac{\Nup}{n_Y}\int\Psi_{\Nup, \Ndown}(x, x_2,...x_{\Nup}, Y)\overline{\Psi_{\Nup, \Ndown}(x', x_2,...,x_{\Nup}, Y)}\d \hat{X}_1 
		\end{equation}
		where $\hat{X}_k = (x_1,... \cancel{x_k},... x_N)$.
	\end{definition}
	Note that we have
	\begin{equation}
		\int n_Y\d Y = 1,~~~0\leq \gamma_Y\leq 1,~~~\tr(\gamma_Y) = \Nup~~~\text{and}~~~\int n_Y \gamma_Y \d Y= \gamma_{\Psi_{\Nup, \Ndown}}^{(1, 0)}.
	\end{equation}	
	Then, we have the following lemma.
	\begin{lemma} Using the notations introduced in Definition \ref{definition_plein_de_trucs}, we have
		\begin{equation}\label{equation_conclusion15}
			\langle \Psi_{\Nup, \Ndown}| H_{\Nup, \Ndown}^{\uparrow}|\Psi_{\Nup, \Ndown}\rangle  \geq N^{-\beta-2/3} (1-s^2 p_F^2)(1 + o(1))\int n_Y \tr(\gamma_Y W_Y)\d Y.
		\end{equation}
	\end{lemma}
	\begin{proof}
		By Lemma \ref{lemma_dyson_our_version}, we have
		\begin{equation}\label{equation_dyson_N}
			H_{\Nup, \Ndown}^{\uparrow} = \sum_{j = 1}^{\Nup}-\hbar^2 \nabla_j \Gamma(p_j)\nabla_j + \frac{1}{2}\sum_{j= 1}^{\Nup}\sum_{k = 1}^{\Ndown} w_N(x_j - y_k)\geq N^{-\beta - 2/3}(1 - s^2 p_F^2)(1 + o(1))\sum_{j= 1 }^{\Nup}W_Y(x_j).
		\end{equation}
		Using the notation $\gamma_Y$, this implies \eqref{equation_conclusion15}.
	\end{proof}
	\begin{remark}
	To estimate $\tr(\gamma_Y W_Y)$, we want to replace $\gamma_Y$ by something simpler that has approximatively the same energy. We set 
	\begin{equation}\label{equation_definition_P_N}
		H^0_{\hbar} = - \hbar^2\Delta + V = \sum_{n\geq1}\lambda_n |e_n\rangle\langle e_n|~~~\text{and}~~~P_{\Nup} = \sum_{n = 1}^{\Nup} |e_n\rangle\langle e_n|.
	\end{equation}
	Then, we have
	\begin{equation}\label{equation_juste_avant_4.3}
		\tr(\gamma_YW_Y) = \tr(P_{\Nup} W_Y^+) + \tr(P_{\Nup}W_Y^-) + \tr\big((\gamma_Y - 1)P_{\Nup}W_YP_{\Nup}\big) + \tr\big(\gamma_Y(W_Y - P_{\Nup} W_Y P_{\Nup})\big).
	\end{equation}
	\hfill $\diamond$
	\end{remark}
	Now, we have to show that $\tr(P_{\Nup} W_Y^+)$ gives the expected interaction energy, and we have to prove that the other terms are error terms that we can control.

	\subsection{Error estimates}\label{subsection_error_estimates}
	Here, we estimate the error terms in (\ref{equation_juste_avant_4.3}). To do so, we prove the following lemma first. 
	\begin{lemma}[Approximation by the free ground state]\label{lemma_graf_seiringer}
		Let $(\Phi_{\Nup, \Ndown})$ satisfying
		\begin{equation}\label{equation_hypothese_lemme_graf_seiringer}
			\langle \Phi_{\Nup, \Ndown}| H_N \Phi_{\Nup, \Ndown}\rangle \leq N E^{\mathrm{TF}}+ O(N^{4/3- \beta}),
		\end{equation}
		then
		\begin{equation}
			\tr\big(\gamma_{\Phi_{\Nup, \Ndown}}^{(1, 0)}(1 - P_{\Nup})\big) = O(N^{17/18} + N^{7/6 - \beta/2}).
		\end{equation}
	\end{lemma}
	This lemma relies on the generalization of Inequality (4.13) of \cite{graf_correlation_1994}, which corresponds to the homogeneous case. In our case, as we do not have an explicite diagonalization of $H^0_{\hbar}$, the computations are slightly more complicated.
	\begin{proof}
		Let $\tau_{\Nup} = \tr\big(\gamma_{\Phi_{\Nup, \Ndown}}^{(1, 0)}(1 - P_{\Nup})\big)$. We define the quantum number of particles and energy by 
		\begin{equation}
			\begin{cases}
				N^{\mathrm{q}}_{\hbar}(\Lambda) = \#\{n\in \N,~\lambda_n \leq \Lambda\}\\
				E^{\mathrm{q}}_{\hbar}(\Lambda) = \sum_{\lambda_n\leq \Lambda}\lambda_n,
			\end{cases}
		\end{equation}
		and the semi-classical number of particles and energy by
		\begin{equation}
			\begin{cases}
				N^{\mathrm{cl}}(\Lambda) = (2\pi)^{-3}\int \mathds{1}_{|p|^2 + V(x)\leq \Lambda}\d x \d p\\
				E^{\mathrm{cl}}(\Lambda) = (2\pi)^{-3}\int (|p|^2 + V(x))\mathds{1}_{|p|^2 + V(x) \leq \Lambda}\d x \d p.
			\end{cases}
		\end{equation}
		For a given $\alpha \in \R$, and for $\Lambda_{\Nup} \in \R$ such that $N^{\mathrm{cl}}(\Lambda_{\Nup}) = \Nup/N$, because of Lemma \ref{lemma_semi_classical_approximation}, we know that 
		\begin{equation}
			N^{\mathrm{q}}_{\hbar}(\Lambda_{\Nup}) = \Nup + O(N^{8/9}),
		\end{equation} 
		and hence, we have
		\begin{align}
			D :=& \Big\langle \Phi_{\Nup, \Ndown}|\sum_{n = 1}^{\Nup} (H^0_{\hbar})_n | \Phi_{\Nup, \Ndown}\Big\rangle -\alpha \tau_{\Nup} \notag\\
			=& \sum_{n\in \N}\lambda_n \langle e_n |\gamma_{\Phi_{\Nup, \Ndown}}^{(1, 0)}e_n\rangle - \alpha \sum_{n\in\N} \mathds{1}_{n\geq N_{\uparrow} + 1} \langle e_n|\gamma_{\Phi_{\Nup, \Ndown}}^{(1, 0)}e_n\rangle\notag \\
			=& \sum_{n\in \N} \Big(\lambda_n - \Lambda_{\Nup} + \alpha\mathds{1}_{\Lambda_{\Nup} \geq \lambda_n}\Big) \langle e_n | \gamma_{\Phi_{\Nup, \Ndown}}^{(1, 0)}e_n\rangle + (\Lambda_{\Nup} - \alpha)N^{\mathrm{q}}_{\hbar}(\Lambda_{\Nup}) + O(N^{8/9})\notag\\
			\geq& \inf\left\{\sum_{n\in \N} \Big(\lambda_n - \Lambda_{\Nup} + \alpha\mathds{1}_{\Lambda_{\Nup} \geq \lambda_n}\Big) M_n,~ 0\leq M_n\leq 1\right\} + (\Lambda_{\Nup} - \alpha) N^{\mathrm{q}}_{\hbar}(\Lambda_{\Nup}) + O(N^{8/9})\notag\\
			\geq& \sum_{n\in \N} \Big(\lambda_n - \Lambda_{\Nup} + \alpha\mathds{1}_{\Lambda_{\Nup} \geq \lambda_n}\Big) \mathds{1}_{(\lambda_n - \Lambda_{\Nup} + \alpha\mathds{1}_{\Lambda_{\Nup} \geq \lambda_n}) \leq 0} + (\Lambda_{\Nup} - \alpha) N^{\mathrm{q}}_{\hbar}(\Lambda_{\Nup}) + O(N^{8/9})\notag\\
			\geq & \sum_{\lambda_n \leq \Lambda_{\Nup} - \alpha}(\lambda_n - \Lambda_{\Nup} + \alpha) + (\Lambda_{\Nup} -\alpha)N^{\mathrm{q}}_{\hbar}(\Lambda_{\Nup}) + O(N^{8/9}) \notag\\
			\geq & E^{\mathrm{q}}_{\hbar}(\Lambda_{\Nup} - \alpha) - (\Lambda_{\Nup} - \alpha)N^{\mathrm{q}}_{\hbar}(\Lambda_{\Nup} - \alpha) + (\Lambda_{\Nup} - \alpha)N^{\mathrm{q}}_{\hbar}(\Lambda_{\Nup}) + O(N^{8/9}).
		\end{align}
		We use Lemma \ref{lemma_semi_classical_approximation} to find
		\begin{equation}
			D \geq N F_N(\alpha) + O(N^{8/9})
		\end{equation}
		for 
		\begin{equation}
			F_{\Nup}(\alpha) = E^{\mathrm{cl}}(\Lambda_{\Nup} - \alpha) - (\Lambda_{\Nup} - \alpha)N^{\mathrm{cl}}(\Lambda_{\Nup} - \alpha) + (\Lambda_{\Nup} - \alpha)N^{\mathrm{cl}}(\Lambda_{\Nup}).
		\end{equation}
		By differentiating $F_N$ at 0, we find
		\begin{equation}
			F_N'(0) = - (E^{\mathrm{cl}})'(\Lambda_{\Nup}) + \Lambda_{\Nup}(N^{\mathrm{cl}})'(\Lambda_{\Nup}) = 0.
		\end{equation}
		Indeed, for $\varepsilon > 0$, we have
		\begin{equation}
			E^{\mathrm{cl}}(\Lambda_{\Nup} + \varepsilon) - E^{\mathrm{cl}}(\Lambda_{\Nup}) = (2\pi)^{-3} \int \big(|p|^2 + V(x)\big)\mathds{1}_{\Lambda_{\Nup} \leq |p|^2 + V(x) \leq \Lambda_{\Nup} + \varepsilon} \d x \d p
		\end{equation}
		and 
		\begin{equation}
			N^{\mathrm{cl}}(\Lambda_{\Nup} + \varepsilon) - N^{\mathrm{cl}}(\Lambda_{\Nup}) = (2\pi)^{-3} \int \mathds{1}_{\Lambda_{\Nup} \leq |p|^2 + V(x) \leq \Lambda_{\Nup} + \varepsilon} \d x \d p,
		\end{equation}
		hence
		\begin{equation}
			\left| \frac{E^{\mathrm{cl}}(\Lambda_{\Nup} + \varepsilon) - E^{\mathrm{cl}}(\Lambda_{\Nup})}{\varepsilon} - \Lambda_{\Nup}\frac{N^{\mathrm{cl}}(\Lambda_{\Nup} + \varepsilon) - N^{\mathrm{cl}}(\Lambda_{\Nup})}{\varepsilon} \right| \leq N^{\mathrm{cl}}(\Lambda_{\Nup} + \varepsilon) - N^{\mathrm{cl}}(\Lambda_{\Nup}) = o(1).
		\end{equation}
		since $N^{\mathrm{cl}}$ is differentiable (see \eqref{H_2}). Therefore, we can write a Taylor expansion of $F_{\Nup}$: there exists $\alpha' \in (0, \alpha)$ such that
		\begin{equation}
			F_{\Nup}(\alpha) = F_N(0) + \frac{\alpha^2}{2}F''_{\Nup}(\alpha') = E^{\mathrm{cl}}(\Lambda_{\Nup}) - \frac{\alpha^2}{2}(N^{\mathrm{cl}})'(\Lambda_{\Nup} - \alpha').
		\end{equation}
		Then, by adding the energy of the spin-down particles and the assumption \eqref{equation_hypothese_lemme_graf_seiringer}, we find
		\begin{equation}
			N E^{\mathrm{TF}} + O(N^{4/3-\beta}) - \alpha \tau_{\Nup} \geq E^{\mathrm{cl}}(\Lambda_{\Nup}) + E^{\mathrm{cl}}(\Lambda_{\Ndown}) - \frac{N\alpha^2}{2}(N^{\mathrm{cl}})'(\Lambda_{\Nup} - \alpha'),
		\end{equation}
		and thus
		\begin{equation}
			\alpha \tau_{\Nup} - \frac{N\alpha^2}{2}(N^{\mathrm{cl}})'(\Lambda_{\Nup} - \alpha') =  O(N^{8/9} + N^{4/3 - \beta}).
		\end{equation}
		Note that for $\alpha$ bounded, $\alpha'$ and thus  $(N^{\mathrm{cl}})'(\Lambda_{\Nup} - \alpha')$ are bounded too. Then, by taking $\alpha$ converging slowly to zero, namely $\alpha = o(1)$ and $\alpha^{-1}(N^{8/9} + N^{4/3 - \beta}) = o(N)$, we find 
		\begin{equation}
			\tau_{\Nup} = \frac{N\alpha}{2}(N^{\mathrm{cl}})'(\Lambda_{\Nup} - \alpha') + \alpha^{-1}O(N^{8/9} + N^{4/3-\beta}) = o(N).
		\end{equation}
		Hence, we can now take $\alpha = c\tau_{\Nup}/N = o(1)$ to find 
		\begin{equation}
			\frac{\tau_{\Nup}^2}{N}\Big(c - \frac{c^2}{2} (N^{\mathrm{cl}})'(\Lambda_{\Nup} - \alpha')\Big) = O(N^{8/9} + N^{4/3-\beta}).
		\end{equation}
		As $(N^{\mathrm{cl}})'(\Lambda_{\Nup} - \alpha')$ is bounded, by choosing $c$ small enough but fixed, we get 
		\begin{equation}
			c - \frac{c^2}{2} (N^{\mathrm{cl}})'(\Lambda_{\Nup} - \alpha') \geq C > 0
		\end{equation}
		and then
		\begin{equation}
			\tau_{\Nup} = O(N^{17/18} + N^{7/6-\beta/2}).
		\end{equation}
	\end{proof}
	
	Now, we can state the main result of this section.
	\begin{proposition}[Bounding the interaction energy using the non-interacting ground state]
		Recall that $W_Y$, $n_Y$ and $\gamma_Y$ are defined in \eqref{equation_definition_W}, \eqref{equation_definition_nY} and \eqref{equation_definition_gammaY}. We have the following inequality for $1 > \delta > 0$
		\begin{equation}\label{equation_errors_estimated}
			\int n_Y \tr(\gamma_Y W_Y)\d Y \geq (1 - \delta) \int n_Y \tr(P_{\Nup} W_Y^+)\d Y - C\delta^{-1}N(N^{-1/18} + N^{1/6 - \beta/2})(R^{-3} + 1/\varepsilon s^2 R) - C\frac{N}{\varepsilon s^2 R}
		\end{equation}
	\end{proposition}
	\begin{proof}
		We divide the proof of this lemma in several parts, using \eqref{equation_juste_avant_4.3^}:\\~\\
			$\bullet$ First, we bound $\tr(P_N W_Y^-)$. Because of (\ref{equation_properties_uR}), we have
			\begin{equation}\label{equation_lemma_errors_1}
				\tr(P_{\Nup}W_Y^-) = \sum_{n= 1 }^{\Nup}\langle e_n| W_Y^- e_n\rangle = \frac{R_w}{\varepsilon}\sum_{n = 1}^{\Nup} \Big\langle e_n \Big| \sum_{k, y_k\in \tilde{Y}}u_R(\cdot - y_k)e_n\Big\rangle  \leq C \frac{\Nup}{\varepsilon s^2 R} \leq C \frac{N}{\varepsilon s^2 R}.
			\end{equation}
			$\bullet$ Then, we bound $\tr\big((\gamma_Y - 1)P_{\Nup}W_YP_{\Nup}\big)$. By positivity of $P_{\Nup}(1 - \gamma_Y)P_{\Nup}$, we have
			\begin{equation}
				\tr\big((\gamma_Y - 1)P_{\Nup}W_YP_{\Nup}\big) \geq - \|W_Y\| \tr\big(P_{\Nup}(1 - \gamma_Y)P_{\Nup}\big) = - \|W_Y\| \tr\big(P_{\Nup}(1 - \gamma_Y)\big).
			\end{equation}
			On the one hand, $\|W_Y\|\leq \|W_Y^+\|_{\infty} + \|W_Y^-\|_{\infty}$. Since the terms in in the sum defining $W_Y^+$ in \eqref{equation_definition_W} have disjoint supports, we have
			\begin{equation}
				\|W_Y^+\|_{\infty} = 4\pi R_w (1 - \varepsilon)\|U_R\|_{\infty} = \frac{3R_w (1 - \varepsilon)}{R^3 - R_0^3} = O(R^{-3}).
			\end{equation}
			The norm $\|W_Y^-\|_{\infty}$ can be bounded easily beacause of the properties of $u_R$ stated in (\ref{equation_properties_uR}):
			\begin{equation}
				\|W_Y^-\|_{\infty} = O\left(\frac{1}{\varepsilon s^2 R}\right).
			\end{equation}
			On the other hand, because of Lemma \ref{lemma_graf_seiringer}, we have 
			\begin{equation}
				\int n_Y \tr\big(P_{\Nup}(1 - \gamma_Y)\big)\d Y = \int n_Y \tr\big(\gamma_Y(1 - P_{\Nup})\big)\d Y = O(N^{17/18} + N^{7/6 - \beta/2}).
			\end{equation}
			Therefore, we find
			\begin{equation}\label{equation_lemma_errors_2}
				\int n_Y\tr\big((\gamma_Y-1)P_{\Nup}W_Y P_{\Nup}\big) \geq - C (N^{17/18} + N^{7/6-\beta/2})(R^{-3}+1/\varepsilon s^2 R).
			\end{equation}
			$\bullet$ Finally, we bound $\tr\big(\gamma_Y(W_Y - P_{\Nup} W_Y P_{\Nup})\big)$. To do so, we separate it in three different parts:
			\begin{multline}
				\tr\big(\gamma_Y(W_Y - P_{\Nup} W_Y P_{\Nup})\big) = \tr\big(\gamma_Y(1 - P_{\Nup})W_Y P_{\Nup}\big) + \tr\big(\gamma_YP_{\Nup} W_Y (1 - P_{\Nup})\big)\\ + \tr\big(\gamma_Y(1 - P_{\Nup}) W_Y (1 - P_{\Nup})\big).
			\end{multline}
			For $A$, $B$ two Hilbert Schmidt operators, we know from Cauchy-Schwarz' inequality that
			\begin{equation}
				|\tr(AB)| \leq \frac{\delta}{2}\tr(AA^*) + \frac{\delta^{-1}}{2}\tr(B^*B).
			\end{equation}
			Applying it to $A = \gamma_Y^{1/2}(1 - P_{\Nup})(W_Y^+)^{1/2}$ and $B = (W_Y^+)^{1/2} P_{\Nup}\gamma_Y^{1/2}$, we find
			\begin{equation}
				\tr\big(\gamma_Y(1 - P_{\Nup})W_Y^+ P_{\Nup}\big) \geq - \frac{\delta}{2}\tr(\gamma_YP_{\Nup}W_Y^+ P_{\Nup}) - \frac{\delta^{-1}}{2}\tr\big(\gamma_Y(1 - P_{\Nup})W_Y^+ (1 - P_{\Nup})\big).
			\end{equation}
			By the same method applied to $A = \gamma_Y^{1/2}(1 - P_{\Nup})(-W_Y^-)^{1/2}$ and $B = (-W_Y^-)^{1/2} P_{\Nup}\gamma_Y^{1/2}$, we find
			\begin{equation}
				-\tr\big(\gamma_Y(1 - P_{\Nup})W_Y^- P_{\Nup}\big) \geq - \frac{\delta}{2}\tr(\gamma_YP_{\Nup}W_Y^- P_{\Nup}) - \frac{\delta^{-1}}{2}\tr\big(\gamma_Y(1 - P_{\Nup})W_Y^- (1 - P_{\Nup})\big).
			\end{equation}
			Moreover, 
			\begin{equation}
				\tr(\gamma_YP_{\Nup}W_Y^+ P_{\Nup}) \leq \tr(P_{\Nup}W_Y^+)
			\end{equation}
			and
			\begin{equation}
				\tr\big(\gamma_Y(1 - P_{\Nup})W_Y^+ (1 - P_{\Nup})\big) \leq \|W_Y^+\|_{\infty}\tr\big((1 - P_{\Nup})\gamma_Y\big),
			\end{equation}
			which implies
			\begin{equation}
				\tr\big(\gamma_Y(1 - P_{\Nup})W_Y P_{\Nup}\big) \geq -\frac{\delta}{2}\tr(P_{\Nup}W_Y^+) -\frac{\delta}{2} \tr(P_{\Nup}W_Y^-) - \frac{\delta^{-1}}{2}\tr\big((1 - P_{\Nup})\gamma_Y\big)(\|W_Y^+\|_{\infty} + \|W_Y^-\|_{\infty}).
			\end{equation}
			It is clear that the same inequality holds for $\tr\big(\gamma_YP_{\Nup} W_Y (1 - P_{\Nup})\big)$, and therefore
			\begin{multline}
				\tr\big(\gamma_Y(W_Y - P_{\Nup} W_Y P_{\Nup})\big) \geq - \delta\tr(P_{\Nup}W_Y^+) - \delta\tr(P_{\Nup}W_Y^-) \\+(1 + \delta^{-1})\tr\big(\gamma_Y(1 - P_{\Nup})\big)(\|W_Y^+\|_{\infty} + \|W_Y^-\|_{\infty}).
			\end{multline}
			Because of Lemma \ref{lemma_graf_seiringer} and of the bounds on $\|W_Y^+\|_{\infty}$ and $\|W_Y^-\|_{\infty}$, we finally find
			\begin{multline}\label{equation_lemma_errors_3}
				\int n_Y\tr\big(\gamma_Y(W_Y - P_{\Nup} W_Y P_{\Nup})\big) \d Y \geq - \delta\tr(P_{\Nup}W_Y^+) - C\frac{N}{\varepsilon s^2 R} \\ -C \delta^{-1} (N^{17/18} + N^{7/6 - \beta/2})(R^{-3} + 1/\varepsilon s² R).
			\end{multline}
		We collect the inequalities (\ref{equation_lemma_errors_1}), (\ref{equation_lemma_errors_2}) and (\ref{equation_lemma_errors_3}) to conclude.
	\end{proof}
	
	\subsection{Convergenge of densities}\label{subsection_convergence_densities}
	Now that we have estimated the errors, we have to check that the main term in (\ref{equation_errors_estimated}) converges to the expected interaction energy. First, we write it in a more convenient way as the interaction between two densities. Then, we show that both densities converge to the Thomas-Fermi density. 
		
	\subsubsection{Main term of the interaction energy}
	
	Recall that $W_Y$ and $n_Y$ are defined in \eqref{equation_definition_W} and \eqref{equation_definition_nY}. By definition, 
	\begin{equation}
		\int n_Y \tr(P_{\Nup}W_Y^+)\d Y = 4 \pi a_w (1-\varepsilon)\int \rho_{P_{\Nup}}(x)\sum_{k=1}^{\Ndown}\int n_Y \mathds{1}_{y_k \in \tilde{Y}}U_R(x - y_k)\d Y\d x.
	\end{equation}	
	Recall that $\tilde{Y} = \{y_j \in Y, ~\forall k\neq j,~|y_j - y_k|\geq 2R\}$; we set
	\begin{equation}\label{equation_definition_rho_tilde}
		\sigma_N^{(0,1)}(y_1) = \int n_Y \mathds{1}_{y_1 \in \tilde{Y}}\d \hat{Y}_1.
	\end{equation}
	The pseudo-dendity $\sigma_N^{(0,1)}$ corresponds to to the density without taking into account the particles that are too close from another. As this is a rare event, the pseudo-density should be close to the density. By symmetry of $n_Y$, we have
	
	\begin{equation}\label{equation_conclusion2}
		\int n_Y \tr(P_{\Nup}W_Y^+)\d Y = 4\pi a_w (1-\varepsilon)\int \rho_{P_{\Nup}}(x) U_R(x - y)\sigma_N^{(0, 1)}(y)\d x \d y.
	\end{equation}
	Besides, if we set $I_R(Y) = \sum_{k=1}^{\Ndown} \mathds{1}_{y_k\notin \tilde{Y}}$, we have
	\begin{equation}
		\int \sigma_N^{(0, 1)}(y)\d y = \Ndown - \sum_{k=1}^{\Ndown}\int n_Y \mathds{1}_{y_k\notin \tilde{Y}}\d Y = \Ndown - \int n_Y I_R(Y)\d Y.
	\end{equation}	
	Then, we use Theorem 5 of \cite{lieb_stability_1988} that we recall here to bound $\int n_Y I_R(Y)\d Y$.
	\begin{lemma}[Bound on the inverse distances]\label{lemma_bound_inverse_distance}
		Let $\delta_k$ be the distance of $y_k$ to its closest neighbour, i.e. $\delta_k = \min_{j\neq k}|y_j - y_k|$. Then, there exists $C$ such that the following inequality for operators on $L^2\big((\mathbb{R}^d)^{\Ndown}\big)$ holds true:
		\begin{equation}
			\sum_{k = 1}^{\Ndown} \frac{1}{\delta_k^2} \leq C \sum_{k = 1}^{\Ndown} -\Delta_k
		\end{equation}
	\end{lemma}	
	Note that
	\begin{equation}
		I_R(Y) = \sum_{k=1}^{\Ndown} \mathds{1}_{\delta_k\leq 2R} \leq \sum_{k=1}^{\Ndown} \frac{(2R)^2}{\delta_k^2} \mathds{1}_{\delta_k\leq 2R} \leq \sum_{k=1}^{\Ndown} \frac{(2R)^2}{\delta_k^2}.
	\end{equation}
	Thus, since $R\ll \hbar$, we can use Lemma \ref{lemma_bound_inverse_distance} to find
	\begin{equation}
		\int n_Y I_R(Y) \d Y \leq C \frac{R^2}{\hbar^2} \Big\langle \Psi_{\Nup, \Ndown}\Big|\sum_{k=1}^{\Ndown} -\hbar^2 \Delta_k \Big|\Psi_{\Nup, \Ndown}\Big\rangle \leq C \frac{R^2}{\hbar^2} \langle \Psi_{\Nup, \Ndown}| H_N \Psi_{\Nup, \Ndown}\rangle\leq C N \frac{R^2}{\hbar^2} = o(N)
	\end{equation}
	Therefore, we have
	\begin{equation}
		\Ndown \geq \int \sigma_N^{(0, 1)}(y) \d y \geq \Ndown - CN \frac{R^2}{\hbar^2},
	\end{equation}
	$\sigma_N^{(0, 1)}$ being defined in \eqref{equation_definition_rho_tilde}. Since $0\leq \sigma_N^{(0, 1)}\leq \rho_{\Psi_{\Nup, \Ndown}}^{(0, 1)}$, we finally find
	\begin{equation}
		\|\sigma_N^{(0, 1)} - \rho_{\Psi_{\Nup, \Ndown}}^{(0, 1)}\|_{L^1} \leq C N \frac{R^2}{\hbar^2}.
	\end{equation}
	
	\subsubsection{Pseudo-density of the minimizer}
	Here, we prove that the pseudo-density of the minimizer, defined in \eqref{equation_definition_rho_tilde}, converges.
	\begin{lemma}[Convergence of the pseudo-density of the minimizer to the Thomas-Fermi density]\label{lemma_convergence_pseudo_density}
		Recall that $\rho^{\mathrm{TF}}$ and $\sigma_N^{(0, 1)}$ are defined in Proposition \ref{proposition_properties_thomas_fermi_density} and \eqref{equation_definition_rho_tilde}. We have 
		\begin{equation}\label{equation_conclusion3}
			\frac{1}{N}\sigma_N^{(0, 1)} \rightharpoonup \frac{1}{2}\rho^{\mathrm{TF}}
		\end{equation}
		weakly in $L^{5/3}$.
	\end{lemma}
	\begin{proof} We show that a mollification of $\sigma_N^{(0, 1)}$ is a minimizing sequence of $\mathcal{E}^{\mathrm{TF}}$, and hence that it converges strongly. Then we deduce that the original sequence converges weakly.~\\~\\
			$\bullet$ Let the coherent state $f_{x, p}^{\hbar}$ be defined by \eqref{equation_definition_fxp}, and let $m_{\uparrow}$ and $m_{\downarrow}$ denote the Husimi functions defined by
			\begin{equation} 
				m_{\downarrow}(x, p) = \Ndown^{-1} \langle f_{x, p}^{\hbar}| \gamma_{\Psi_{\Nup, \Ndown}}^{(0, 1)} f_{x, p}^{\hbar}\rangle. 
			\end{equation}
			Let
			\begin{equation}
				\gamma_{\sqrt{\hbar}}= \hbar^{-3/2}f(\cdot/\sqrt{\hbar})
			\end{equation}
			with $f\in C^{\infty}_c(\mathbb{R}^3)$ a real positive and radial function, then, we have (see the proof of Lemma \ref{lemma_semi_classical_approximation} in the Appendix for more details)
			\begin{multline}
				\tr(-\hbar^2 \Delta \gamma_{\Psi_{\Nup, \Ndown}}^{(0, 1)})+\tr(V\gamma_{\Psi_{\Nup, \Ndown}}^{(0, 1)}) = \Ndown \iint |p|^2 m_{\downarrow}(x, p)\d x \d p + \Ndown \iint V(x) m_{\downarrow}(x, p)\d x \d p + o(\Ndown)\\
				\geq \Ndown c^{\mathrm{TF}}\int \rho_{m_{\downarrow}}(x)^{5/3}\d x + \Ndown\int V(x) \rho_{m_{\downarrow}}(x)\d x + o(\Ndown)
			\end{multline}
			with 
			\begin{equation} 
				\rho_{m_{\downarrow}} = N^{-1}\rho_{\Psi_{\Nup, \Ndown}}^{(0, 1)} * \gamma_{\sqrt{\hbar}}. 
			\end{equation} 
			By the same computation for the spin-up particles, and because of \eqref{equation_supposition_etat_test} and the positivity of the interaction $w$ we find
			\begin{equation}
				\mathcal{E}_{(2)}^{\mathrm{TF}}[\rho_{m_{\uparrow}}, \rho_{m_{\downarrow}}] \leq E^{\mathrm{TF}} + o(1).
			\end{equation}
			Now, we set 
			\begin{equation} 
				\bar{\sigma}_N^{(0, 1)} := N^{-1}\sigma_N^{(0, 1)}*\gamma_{\sqrt{\hbar}},
			\end{equation}
			and as 
			\begin{equation}
				\sigma_N^{(0, 1)} \leq \rho_{\Psi_{\Nup, \Ndown}}^{(0, 1)},
			\end{equation}
			we have
			\begin{equation}
				\bar{\sigma}_N^{(0, 1)} \leq N^{-1}\rho_{\Psi_{\Nup, \Ndown}}^{(0, 1)} * \gamma_{\sqrt{\hbar}} = \rho_{m_{\downarrow}}.
			\end{equation}
			Then it is clear that
			\begin{equation}
				\mathcal{E}_{(2)}^{\mathrm{TF}}[\bar{\sigma}_N^{(1, 0)}, \bar{\sigma}_N^{(0, 1)}]\leq E^{\mathrm{TF}} + o(1).
			\end{equation}
			As 
			\begin{equation} 
				\int (\bar{\sigma}_N^{(1, 0)} + \bar{\sigma}_N^{(0, 1)}) = N^{-1}\int (\sigma_N^{(1, 0)} + \sigma_N^{(0, 1)}) = 1 + o(1), 
			\end{equation} 
			the couple $\Big(\frac{\bar{\sigma}_N^{(1, 0)}}{\|\bar{\sigma}_N^{(1, 0)}\| + \|\bar{\sigma}_N^{(0, 1)}\|}, \frac{\bar{\sigma}_N^{(0, 1)}}{\|\bar{\sigma}_N^{(1, 0)}\| + \|\bar{\sigma}_N^{(0, 1)}\|}\Big)$ minimizes $\mathcal{E}_{(2)}^{\mathrm{TF}}$ and hence by Lemma \ref{lemme_convergence_suites_minimisantes}, we have convergence towards $\rho^{\mathrm{TF}}$:
			\begin{equation}
				\frac{\bar{\sigma}_N^{(0, 1)}}{\|\bar{\sigma}_N^{(1, 0)}\|_{L^1} + \|\bar{\sigma}_N^{(0, 1)}\|_{L^1}} \to \frac{1}{2}\rho^{\mathrm{TF}}
			\end{equation}
			strongly in $L^{5/3}$. Finally, since $\|\bar{\sigma}_N^{(1, 0)}\|_{L^1} + \|\bar{\sigma}_N^{(0, 1)}\|_{L^1} \to 1$, we have
			\begin{equation}\label{equation_convergence_forte_pseudodensite}
				\bar{\sigma}_N^{(0, 1)} = N^{-1}\sigma_N^{(0, 1)} * \gamma_{\sqrt{\hbar}} \to \frac{1}{2}\rho^{\mathrm{TF}}.
			\end{equation}
			$\bullet$ Let $\varphi \in L^{5/2}(\R^3)$, we have
			\begin{equation}
				\left|\int \big(2N^{-1} \sigma_N^{(0, 1)} - \rho^{\mathrm{TF}}\big)\varphi\right| \leq \big\|2 \bar{\sigma}_N^{(0, 1)} * \gamma_{\sqrt{\hbar}} -  \rho^{\mathrm{TF}}\big\|_{L^{5/3}}\|\varphi\|_{L^{5/2}} + \|2N^{-1}\sigma_N^{(0, 1)}\|_{L^{5/3}}\big\|\varphi - \varphi * \gamma_{\sqrt{\hbar}}\big\|_{L^{5/2}}.
			\end{equation}
			Then by \eqref{equation_convergence_forte_pseudodensite} and e.g. Theorem 4.22 of \cite{brezis_functional_2011}), we have
			\begin{equation}
				\int \Big(\frac{1}{N}\sigma_N^{(0, 1)} - \frac{1}{2}\rho^{\mathrm{TF}}\Big)\varphi \to 0
			\end{equation}
			which means that $\frac{1}{N}\sigma_N^{(0, 1)}$ converges weakly to $\frac{1}{2}\rho^{\mathrm{TF}}$ in $L^{5/3}$.
	\end{proof}
	
	\begin{remark}\label{remark_convergence_ratio_Ndown_N}
		In the proof, we have proved that $\big(N^{-1}\rho_{\Psi_{\Nup, \Ndown}}^{(1, 0)} * \gamma_{\sqrt{\hbar}}, N^{-1}\rho_{\Psi_{\Nup, \Ndown}}^{(0, 1)} * \gamma_{\sqrt{\hbar}}\big)$ is a minimizing sequence for $\mathcal{E}_{(2)}^{\mathrm{TF}}$. This implies directly the convergence
		\begin{equation}
			N^{-1}\rho_{\Psi_{\Nup, \Ndown}}^{(0, 1)} * \gamma_{\sqrt{\hbar}} \to \frac{1}{2}\rho^{\mathrm{TF}}
		\end{equation}
		strongly in $L^{5/3}$ by Lemma \ref{lemme_convergence_suites_minimisantes}. Moreover $N^{-1}\int V\rho_{\Psi_{N_{\uparrow}, N_{\downarrow}}}^{(0, 1)}$ is bounded and thus $(N^{-1}\rho_{\Psi_{N_{\uparrow}, N_{\downarrow}}}^{(0, 1)})$ is a tight sequence of measures. Therefore, there exists a measure $\mu$ such that $(N^{-1}\rho_{\Psi_{N_{\uparrow}, N_{\downarrow}}}^{(0, 1)})$ converges weakly to $\mu$ as a sequence of measures, and
		\begin{equation}
			N^{-1}\int \rho_{\Psi_{N_{\uparrow}, N_{\downarrow}}}^{(0, 1)} \to \int \d\mu.
		\end{equation}
		In particular, $(\rho_{\Psi_{N_{\uparrow}, N_{\downarrow}}}^{(0, 1)})$ converges to $\mu$ as a distribution, and by uniqueness of the limit, we have $\mu = \rho^{\mathrm{TF}}$ and then
		\begin{equation}
			\frac{N_{\downarrow}}{N} = \frac{1}{N}\int \rho_{\Psi_{\Nup, \Ndown}}^{(0, 1)} \to \frac{1}{2}\int \rho^{\mathrm{TF}} = \frac{1}{2}.
		\end{equation}
		\hfill $\diamond$
	\end{remark}
	
	\subsubsection{Density of the free ground state}
	
	First, we recall why the convergence of the density of the projector on the free ground state $\rho_{P_{\Nup}}$ defined by \eqref{equation_definition_P_N} implies the convergence of its convolution with $U_R$.
	\begin{lemma}[Regularization of a converging sequence]
		Let $p \in (1, \infty)$, we assume that the sequence $(\rho_{P_{\Nup}}/\Nup)$ converges to $\rho^{\mathrm{TF}}$ in $L^p$, then
		\begin{equation}\label{equation_conclusion4}
			\frac{1}{\Nup}\rho_{P_{\Nup}}*U_R \to \rho^{\mathrm{TF}}
		\end{equation}
		with $P_N$, $U_R$ and $\rho^{\mathrm{TF}}$ defined in \eqref{equation_definition_P_N}, \eqref{equation_definition_UR} and Proposition \ref{proposition_properties_thomas_fermi_density}, and with $\hbar \geq R > R_0 \geq N^{-\beta}R_w$.
	\end{lemma}
	\begin{proof}
		We know that
		\begin{equation}
			\|\rho_{P_{\Nup}}*U_R - \Nup \rho^{\mathrm{TF}}\|_{L^p} \leq \|(\rho_{P_{\Nup}}- N \rho^{\mathrm{TF}})*U_R\|_{L^p} + N \|\rho^{\mathrm{TF}}*U_R - \rho^{\mathrm{TF}}\|_{L^p}.
		\end{equation}
		By Jensen's inequality, 
		\begin{equation}\label{equation_inegalité_Jensen}
			\Big((\rho_{P_{\Nup}} - N \rho^{\mathrm{TF}})*U_R\Big)^p = \left(\int (\rho_{P_{\Nup}} - N \rho^{\mathrm{TF}})(y) U_R(\cdot-y)\d y\right)^p \leq \int (\rho_{P_{\Nup}} - N \rho^{\mathrm{TF}})^p(y) U_R(\cdot-y)\d y.
		\end{equation}
		By integrating \eqref{equation_inegalité_Jensen}, we find
		\begin{equation}
			\|(\rho_{P_{\Nup}} - N \rho^{\mathrm{TF}})*U_R\|_{L^p} \leq \|\rho_{P_{\Nup}} - N \rho^{\mathrm{TF}}\|_{L^p},
		\end{equation}
		which together with the convergence of $\rho^{\mathrm{TF}}*U_R$ to $\rho^{\mathrm{TF}}$ concludes the proof.
	\end{proof}
	
	Now, only the convergence of the normalized density of the free ground state $\frac{1}{\Nup}\rho_{P_{\Nup}}$ remains to be shown. As the normalized pseudo-density $\tilde{\rho}/N$ converges weakly in $L^{5/3}$, we have to prove the strong convergence of the density of the projector in $L^{5/2}$ to have convergence of the right-hande side of \eqref{equation_conclusion2}.
	\begin{lemma}[Convergence of the density of the free ground state]\label{lemma_convergence_densite_projecteur}
		Recall that $P_{\Nup}$ and $\rho^{\mathrm{TF}}$ are defined in \eqref{equation_definition_P_N} and Proposition \ref{proposition_properties_thomas_fermi_density}. We have
		\begin{equation}\label{equation_conclusion5}
			\frac{1}{N}\rho_{P_{\Nup}} \to \frac{1}{2}\rho^{\mathrm{TF}}
		\end{equation}
		strongly in $L^{5/2}$.
	\end{lemma}
	\begin{remark}
		This result may be proved with arguments from \cite{cardenas_commutator_2025}. We present here a different proof.
		\hfill $\diamond$
	\end{remark}
	\begin{proof}The goal here is to prove that $\Check{\rho}_{\Nup}$ is bounded in the Sobolev space $W^{s, 5/2}$ (with $s>0$) in order to obtain strong convergence in $L^{5/2}$ by a compact injection. To do so, we show that $\Check{\rho}_{\Nup}$ is bounded in $L^3$ and $\dot{W}^{1, 1}$ and we conclude using a Gagliardo-Nirenberg interpolation inequality.\\
			$\bullet$ In order to prove that $\Check{\rho}_{\Nup}$ is bounded in $L^3$, we want to bound $\tr\big((-\Delta)^3P_{N_{\uparrow}}\big)$. Before that, we have to bound $\tr\big((-\Delta)^2P_{N_{\uparrow}}\big)$. We denote $p = -i\hbar\nabla$ and we have $H^0_{\hbar} = p^2 + V$. Then,
			\begin{equation}
				(H^0_{\hbar})^2 = p^4 + p^2 V + Vp^2 + V^2 = p^4 + V^2 + pVp + [p, [p, V]].
			\end{equation}
			Since $V \geq 0$, we immediately have $pVp \geq 0$. On the other hand, we find
			\begin{equation}
				[p, [p, V]] = -\hbar^2 \Delta V.
			\end{equation}
			As $|\Delta V| \leq C V^2$, for $N$ large enough (and hence $\hbar$ small enough), we have 
			\begin{equation}
				\frac{1}{2}V^2 + [p, [p, V]] \geq 0
			\end{equation}
			and thus
			\begin{equation}
				(H^0_{\hbar})^2 \geq \frac{1}{2}V^2 + p^4 = \frac{1}{2}V^2 + (-\hbar^2\Delta)^2.
			\end{equation}
			By definition, $P_{\Nup}$ minimizes $\tr(H^0_{\hbar}\gamma)$ for $0\leq \gamma \leq 1$ and $\tr(\gamma) = \Nup$, and it also minimizes $\tr\big((H^0_{\hbar})^2 \gamma\big)$ with the samed constraints on $\gamma$. Then, we have (for $N$ large enough)
			\begin{equation}\label{equation_bound_laplacian_squarred}
				\tr(V^2 P_{\Nup}) + \tr(\hbar^4 (-\Delta)^2 P_{\Nup})\leq \tr\big((H^0_{\hbar})^2 P_{\Nup}\big) \leq C N.
			\end{equation}
			$\bullet$ Now, we want to prove a bound similar to (\ref{equation_bound_laplacian_squarred}) for $(-\hbar^2\Delta)^3$ in order to bound $\check{\rho}_{\Nup}$ in $L^3$. We have
			\begin{equation}
				(H^0_{\hbar})^3 \geq H^0_{\hbar} p^2 H^0_{\hbar} = p^6 + p^4 V + Vp^4 + V p^2 V \geq p^6 + p^4V + V p^4.
			\end{equation}
			Since $p^4 V = p^2 V p^2 + p^2[p^2, V] \geq p^2 [p^2, V]$ and $Vp^4 \geq -[p^2, V] p^2$, we have
			\begin{equation}
				(H^0_{\hbar})^3 \geq p^6 + [p^2, [p^2, V]] = p^6 + \hbar^4 [\Delta,[\Delta, V]].
			\end{equation}
			We have to compute the double commutator above in order to bound $p^6$ by $(H^0_{\hbar})^3$. For $\psi \in \Hcal$
			\begin{equation}
				[\Delta, V] \psi = \Delta (V\psi) - V \Delta \psi = (\Delta V)\psi + 2 \nabla V.\nabla\psi.
			\end{equation}
			Then,
			\begin{equation}
				[\Delta,[\Delta, V]]\psi = (\Delta^2 V)\psi + 4 \nabla \Delta V .\nabla \psi + 2 \nabla^{\otimes 2}V {:} \nabla^{\otimes 2}\psi,
			\end{equation}
			with
			\begin{equation}
				\nabla^{\otimes 2}V = (\nabla_{ij}V)_{i, j}~~~\text{and}~~~\nabla^{\otimes 2}V {:} \nabla^{\otimes 2}\psi = \sum_{i, j}\overline{\partial_{ij}V}\partial_{ij}\psi.
			\end{equation}
			Note that 
			\begin{equation}
				\nabla\Delta V.\nabla \geq -\frac{1}{2}(-\Delta) - \frac{1}{2}|\nabla \Delta V|^2
			\end{equation}
			and
			\begin{equation}
				\nabla^{\otimes 2}V{:}\nabla^{\otimes 2} \geq -\frac{\hbar^{-2}}{2}|\nabla^{\otimes 2}V|^2 - \frac{\hbar^2}{2}|-\nabla^{\otimes 2}|^2 = -\frac{\hbar^{-2}}{2}|\nabla^{\otimes 2}V|^2 - \frac{\hbar^2}{2}(-\Delta)^2.
			\end{equation}
			Therefore, because of (\ref{H_1}), we find
			\begin{equation}
				\tr\Big([\Delta,[\Delta, V]]P_{\Nup}\Big)\geq -2\tr((-\Delta)P_{\Nup}) - C(3 + \hbar^{-2})\tr(V^2P_{\Nup}) - \hbar^2 \tr((-\Delta)^2 P_{\Nup})
			\end{equation}
			and because of (\ref{equation_bound_laplacian_squarred}),
			\begin{equation}
				\tr\Big([\Delta,[\Delta, V]]P_{\Nup}\Big)\geq - C\hbar^{-2}N.
			\end{equation}
			Thus, we finally have
			\begin{equation}
				C N \geq \tr\big((H^0_{\hbar})^3 P_{\Nup})\geq \tr(\hbar^6 (-\Delta)^3P_{\Nup}) + O(\hbar^2 N).
			\end{equation}
			Then, the Lieb-Thirring inequality \cite[Eq. (3.3)]{daubechies_uncertainty_1983} gives
			\begin{equation}
				\tr(\hbar^6 (-\Delta)^3P_{\Nup}) \geq C N^{-2}\int (\rho_{P_{\Nup}})^3,
			\end{equation}
			and hence
			\begin{equation}
				\int \rho_{P_{\Nup}}^3 \leq C N^3.
			\end{equation}
			$\bullet$ Now that that we know that $N^{-1}\rho_{P_{\Nup}}$ is bounded in $L^3$, we want to bound it in $\dot{W}^{1, 1}$ to apply the Gagliardo-Nirenberg inequality. Thanks to Theorem 1.2. of \cite{fournais_optimal_2020}, we know that the trace norm of $[\nabla, P_{\Nup}]$ is bounded:
			\begin{equation}\label{equation_borne_commutateur_gradient}
				\|[\nabla, P_{\Nup}]\|_{\sigma^1}\leq C N.
			\end{equation}
			In particular, since for any $x\in \R^3$
			\begin{equation}
				[\nabla, P_N](x, x) = \nabla \rho_{P_{\Nup}}(x),
			\end{equation} 
			\eqref{equation_borne_commutateur_gradient} implies that
			\begin{equation}
				\sup_{\|\varphi\|_{L^{\infty}} = 1}\left| \int \varphi \nabla \rho_{P_{\Nup}}\right|\leq C N.
			\end{equation} 
			Therefore, $\rho_{P_{\Nup}}\in \dot{W}^{1, 1}$ and $\|\nabla\rho_{P_{\Nup}}\|_{L^1}\leq C N$. Therefore, $N^{-1} \nabla \rho_{\Nup}$ is bounded. We use Gagliardo Nirenberg inequality to find
			\begin{equation}
				\|N^{-1}\rho_{P_{\Nup}}\|_{\dot{W}^{1/10, 5/2}} \leq C \|N^{-1}\rho_{P_{\Nup}}\|_{\dot{W}^{1, 1}}^{1/10} \|N^{-1}\rho_{P_{\Nup}}\|_{L^{3}}^{9/10}.
			\end{equation}
			Moreover, $N^{-1}\rho_{P_{\Nup}}$ is bounded in $L^1$ and in $L^3$, and thus in $L^{5/2}$. Therefore, $N^{-1}\rho_{P_{\Nup}}$ is bounded in $W^{1/10, 5/2}$, and converges strongly  up to extraction to a certain $\rho_{\infty}$ in $L^{5/2}$. Furthermore, as $\Nup/N \to 1/2$ (see Remark \ref{remark_convergence_ratio_Ndown_N}), $P_{\Nup} \otimes P_{\Ndown}$ minimizes approximatively $H_N$ (up to a $O(N)$), and hence $N^{-1}\rho_{P_{\Nup}}$ converges weakly to $\frac{1}{2}\rho^{\mathrm{TF}}$. Therefore, we find $\rho_{\infty} = \frac{1}{2}\rho^{\mathrm{TF}}$ and since it is the only possible limit, we have convergence of the whole sequence $(N^{-1}\rho_{P_{\Nup}})$ to $\frac{1}{2}\rho^{\mathrm{TF}}$, strongly in $L^{5/2}$.
	\end{proof}
	
	\subsection{Conlusion}\label{subsection_conclusion}
	From (\ref{equation_conclusion3}), (\ref{equation_conclusion4}) and (\ref{equation_conclusion5}), we know that $\sigma_N^{(0,1)}/N$ converges weakly in $L^{5/2}$ and $\rho_{P_{\Nup}}*U_R/N$ strongly in $L^{5/3}$. Therefore,
	\begin{equation}\label{equation_juste_avant_conclusion}
		\int \rho_{P_{\Nup}}(x)U_R(x-y)\sigma_N^{(0, 1)}(y)\d x \d y \to \int \big(\rho^{\mathrm{TF}}\big)^2.
	\end{equation}
	Then, putting together (\ref{equation_conclusion1}), (\ref{equation_conclusion15}), (\ref{equation_errors_estimated}), (\ref{equation_conclusion2}) and \eqref{equation_juste_avant_conclusion}, we find
	\begin{equation}
		\langle\Psi_{\Nup, \Ndown}|H_N| \Psi_{\Nup, \Ndown}\rangle \geq N E^{\mathrm{TF}} + 2\pi a_w(1-\varepsilon -\delta - s^2p_{\mathrm{F}}^2) N^{4/3-\beta} \int \big(\rho^{\mathrm{TF}}\big)^2 + C f(N)
	\end{equation}
	with
	\begin{equation}
		f(N) = N^{5/6} + p_{\mathrm{F}}^{-2}N + \delta^{-1}N^{1/3 - \beta} \Big(N^{-1/18} + N^{1/6 - \beta/2}\Big) \left(R^{-3} + \frac{1}{\varepsilon s^2 R}\right) + \frac{N^{1/3 - \beta}}{\varepsilon s^2 R}.
	\end{equation}	
	We take 
	\begin{equation} \label{equation_choix_constantes}
		\begin{cases} 
			\delta = \varepsilon
			p_{\mathrm{F}}^{-2} = \varepsilon N^{1/3-\beta}
			s^2 = \varepsilon^2 N^{1/3-\beta}
			R = \varepsilon \hbar,
		\end{cases}
	\end{equation}
	and we find
	\begin{align}
		f(N) &= N^{5/6} + \varepsilon N^{4/3-\beta} + \varepsilon^{-1}N^{1/3-\beta}(N^{-1/18} + N^{1/6-\beta/2})(\varepsilon^{-3}N + \varepsilon^{-4}N^{\beta}) + \varepsilon^{-4} N^{1/3}\notag\\
		&\leq N^{5/6} + \varepsilon N^{4/3-\beta} + \varepsilon^{-4}N^{4/3 - \beta} (N^{-1/18} + N^{1/6-\beta/2}) + \varepsilon^{-4} N^{1/3}.
	\end{align}	
	Note that $N^{5/6} = o(N^{4/3-\beta})$ if and only if $\beta < 1/2$. Furthermore, we impose 
	\begin{equation}
		1 \gg \varepsilon \gg (N^{-1/18} + N^{1/6-\beta/2})^{1/4},
	\end{equation}	
	which implies
	\begin{equation}
		f(N) = o(N^{4/3-\beta})
	\end{equation}
	and
	\begin{equation}
		\langle\Psi_{\Nup, \Ndown}|H_N |\Psi_{\Nup, \Ndown}\rangle \geq N E^{\mathrm{TF}} + 2\pi a_w(1 + o(1)) N^{4/3-\beta} \int \big(\rho^{\mathrm{TF}}\big)^2 + O(N^{4/3-\beta}).
	\end{equation}
	\appendix
	\section{Appendix: Some semi-classical results}
	
	In this section, we recall some well known definitions and results in semi-classical analysis (see Section 2.1 of \cite{fournais_semi-classical_2018} for more detail and proofs). Moreover, we prove that the semi-classical energy approximates the quantum energy.
	
	\begin{definition}
		We choose $\hbar_x, ~\hbar_p$ such that $\hbar_x\hbar_p = \hbar^2$ and $\hbar_x,~\hbar_p \ll 1$. We fix $f\in C^{\infty}_c(\R)$ a real, positive and radial function. For $y,q\in \R^3$, we set 
		\begin{equation}
			\begin{cases}
				f^{\hbar}(y) = \hbar_x^{-3/4}f(y/\sqrt{\hbar_x})\\
				g^{\hbar}(q) = \hbar_p^{-3/4}\hat{f}(q/\sqrt{\hbar_p})
			\end{cases}
		\end{equation}
		with
		\begin{equation}
			\hat{f}(q) = (2\pi)^{-3/2}\int f(y) e^{-iq.y}\d y
		\end{equation}
		For $x, y, p, q\in \R^3$, we set
		\begin{equation}\label{equation_definition_fxp}
			f_{x,p}^{\hbar} (y) = \hbar_x^{-3/4}f\left(\frac{y-x}{\sqrt{\hbar_x}}\right)e^{ip.y/\hbar}.
		\end{equation}
		For $\gamma = \sum \alpha_n|u_n\rangle\langle u_n|$ a positive trace-class operator, we define respectively the space and momentum densities of $\gamma$ by
		\begin{equation}
			\rho_{\gamma}= \sum \alpha_n |u_n|^2~~~\mathrm{and}~~~t_{\gamma} = \sum_n \alpha_n\big|\mathcal{F}_{\hbar}[u_n]\big|^2.
		\end{equation}
	\end{definition}
	We give some well-known results on $f_{x, p}^{\hbar}$, $\rho_{\gamma}$ and $t_{\gamma}$.
	\begin{proposition}
		We have 
		\begin{equation}\label{equation_resolution_identity_f}
			\int |f_{x, p}^{\hbar}\rangle\langle f_{x, p}^{\hbar}| \d x \d p = (2\pi\hbar)^3 = \frac{(2\pi)^3}{N}.
		\end{equation}
		For a given $\gamma$, we have 
		\begin{equation}\label{equation_integral_density_trace}
			\int \rho_{\gamma}(x)\d x = \int t_{\gamma}(p)\d p = \tr(\gamma).
		\end{equation}
	\end{proposition}	
	We recall that 
	\begin{equation}
		H^0_{\hbar} = -\hbar^2 \Delta + V = \sum_{n \geq1}\lambda_n |e_n\rangle\langle e_n| ~~~~~ \text{and} ~~~~~ P_N = \sum_{n=1}^N |e_n\rangle\langle e_n |.
	\end{equation}
	Now, we recall notation introduced in Section \ref{subsection_error_estimates}. For a given $\Lambda \in \R_+$, we have
	\begin{equation}
		\begin{cases}
			N^{\mathrm{q}}_{\hbar}(\Lambda) = \#\{n\in \N,~\lambda_n \leq \Lambda\}\\
			E^{\mathrm{q}}_{\hbar}(\Lambda) = \sum_{\lambda_n\leq \Lambda}\lambda_n\\
			N^{\mathrm{cl}}(\Lambda) = (2\pi)^{-3}\int \mathds{1}_{|p|^2 + V(x)\leq \Lambda}\d x \d p\\
			E^{\mathrm{cl}}(\Lambda) = (2\pi)^{-3}\int (|p|^2 + V(x))\mathds{1}_{|p|^2 + V(x) \leq \Lambda}\d x \d p.
		\end{cases}
	\end{equation}	
	We want to prove the following result:
	\begin{lemma}[semi-classical approximation of the energy]\label{lemma_semi_classical_approximation}
		Let $\Lambda \in \R_ +$, we have
		\begin{equation}\label{equation_semi_classique_energie}
			N E^{\mathrm{cl}}(\Lambda) = E^{\mathrm{q}}_{\hbar}(\Lambda) + O(N^{8/9}).
		\end{equation}
		and
		\begin{equation}\label{equation_semi_classique_nombre}
			N N^{\mathrm{cl}}(\Lambda) = N^{\mathrm{q}}_{\hbar}(\Lambda) + O(N^{8/9})
		\end{equation}
	\end{lemma}
	\begin{remark}
		If we choose $\Lambda$ such that $N^{\mathrm{cl}}(\Lambda) = 1$, we find 
		\begin{equation}
			N^{\mathrm{q}}(\Lambda) = N + o(N^{8/9}).
		\end{equation}
		\hfill $\diamond$
	\end{remark}
	\begin{proof}
		We fix $\Lambda \in \mathbb{R}_+$. Let $\tilde{E}^{\mathrm{q}}_{\hbar}(\Lambda)$ and $\tilde{E}^{\mathrm{cl}}(\Lambda)$ be the variational energies defined respectively by
		\begin{equation}\label{equation_quantum_energy_variational}
			\tilde{E}^{\mathrm{q}}_{\hbar}(\Lambda) = \min \big\{\tr\big((H^0_{\hbar}-\Lambda)\gamma\big),~ 0\leq \gamma\leq 1\big\} = \tr\big((H^0_{\hbar} - \Lambda)P_N\big) = E^{\mathrm{q}}_{\hbar}(\Lambda) - \Lambda N^{\mathrm{cl}}(\Lambda)
		\end{equation}
		and
		\begin{equation}\label{equation_semiclassical_energy_variational}
			\tilde{E}^{\mathrm{cl}}(\Lambda) = \min \left\{\frac{1}{(2\pi)^3}\iint \big(|p|^2 + V(x) - \Lambda\big)m(x, p)\d x\d p,~ 0\leq m\leq 1 \right\} = E^{\mathrm{cl}}(\Lambda) - \Lambda N^{\mathrm{cl}}(\Lambda).
		\end{equation}
		First, we show the counterpart of (\ref{equation_semi_classique_energie}) for $\tilde{E}^{\mathrm{q}}_{\hbar}$ and $\tilde{E}^{\mathrm{cl}}$; to do so, we show an upper bound and a lower bound.\\~\\
		\textbf{Upper bound for the variational energies.} Let $m_0(x, p) = \mathds{1}_{|p|^2 + V(x) \leq \Lambda}$ be the minimizer of the semi-classical variational problem (\ref{equation_semiclassical_energy_variational}), we define the test one-body density matrix $\gamma$ by
		\begin{equation}
			\gamma = \frac{N}{(2\pi)^3}\iint m_0(x, p) |f_{x, p}^{\hbar}\rangle \langle f_{x, p}^{\hbar}|\d x \d p.
		\end{equation}
		Because of (\ref{equation_resolution_identity_f}), we know that $0 \leq \gamma \leq 1$, and hence
		\begin{equation}
			\tilde{E}^{\mathrm{q}}_{\hbar}(\Lambda) \leq \tr\big((H^0_{\hbar} - \Lambda)\gamma\big) = \frac{N}{(2\pi)^3}\iint m_0(x, p)\langle f_{x, p}^{\hbar}|(H^0_{\hbar} - \Lambda)f_{x, p}^{\hbar}\rangle.
		\end{equation}
		Now, we compute separately the kinetic and potentiel energies. On the one hand, 
		\begin{align}
			\int |\nabla f_{x, p}^{\hbar}|^2 &= \hbar_x^{-3/2} \int \left| \nabla_y\left( f\left(\frac{y - x}{\sqrt{h_x}}\right)e^{ipy/\hbar}\right)\right|^2 \d y \nonumber\\
			&= \hbar_x^{-3/2}\int \left|\hbar_x^{-1/2}\nabla f\left(\frac{y - x}{\sqrt{\hbar_x}}\right)e^{ipy/\hbar} + \frac{ip}{\hbar}f\left(\frac{y-x}{\sqrt{\hbar_x}}\right)e^{ipy/\hbar}\right|^2 \d y \nonumber\\
			&= \hbar_x^{-3/2} \int \hbar_x^{-1} \left|\nabla f\left(\frac{y - x}{\sqrt{\hbar_x}}\right)\right|^2\d y + \hbar_x^{-3/2}\int  \hbar^{-2}|p|^2 \left| f\left(\frac{y - x}{\sqrt{\hbar_x}}\right)\right|^2\d y\nonumber\\
			&= \hbar_x^{-1} \|\nabla f\|_{L^2}^2 + \hbar^{-2}|p|^2.
		\end{align}		
		On the other hand,
		\begin{equation}
			\hbar_x^{-3/2} \int V(y) \left|f\left(\frac{y - x}{\sqrt{\hbar_x}}\right)\right|^2 \d y= \int |f(y)|^2 V(x + \sqrt{\hbar_x}y) \d y = V(x) + O\big(\|\nabla V\|_{L^{\infty}(\{V \leq \Lambda + 1\})}\sqrt{\hbar_x}\big).
		\end{equation}		
		Therefore,
		
		\begin{equation}
			\tilde{E}^{\mathrm{q}}_{\hbar}(\Lambda)\leq N\big(\tilde{E}^{\mathrm{cl}}(\Lambda)+O(\hbar_p + \sqrt{\hbar_x})\big).
		\end{equation}
		Finally, we optimize by choosing $\sqrt{\hbar_x} = \hbar_p = \hbar^{2/3}$, which gives
		\begin{equation}\label{equation_variational_energy_upper_bound}
			\tilde{E}^{\mathrm{q}}_{\hbar}(\Lambda)\leq N\tilde{E}^{\mathrm{cl}} + O(N^{7/9}).
		\end{equation}
		\textbf{Lower bound for the variational energies.} Let $\gamma_0 = \sum_{\lambda_n\leq \Lambda}|e_n\rangle\langle e_n|$ be the minimizer of the quantum variational problem (\ref{equation_quantum_energy_variational}), we define the Husimi function associated to $\gamma_0$ by
		\begin{equation}
			m(x, p) = \langle f_{x, p}^{\hbar}| \gamma_0| f_{x, p}^{\hbar}\rangle.
		\end{equation}
		As $0\leq \gamma_0\leq 1$ and $\|f_{x, p}^{\hbar}\|^2 = 1$, we have
		\begin{equation}
			\tilde{E}^{cl}(\Lambda) \leq \frac{1}{(2\pi)^d}\iint \big(|p|^2 + V(x) - \Lambda\big) m(x, p)\d x \d p = \frac{1}{(2\pi)^d} \iint \big(|p|^2 + V(x)-\Lambda\big) \langle f_{x, p}^{\hbar} | \gamma_0 | f_{x, p}^{\hbar}\rangle\d x \d p.
		\end{equation}
		Once again, we compute separately the kinetic and potential energies. On the one hand,
		\begin{align}\label{equation_kinetic_energy_husimi}
			\frac{1}{(2\pi \hbar)^3} \int |p|^2 \langle f_{x, p}^{\hbar} | \gamma_0 | f_{x,p}^{\hbar}\rangle\d x \d p &= \int |p|^2 \big(t_{\gamma_0}*|g^{\hbar}|^2\big)(p) \d p\nonumber\\
			&= \iint |p-q|^2  t_{\gamma_0}(p) |g^{\hbar}(q)|^2 \d q \d p \nonumber\\
			&= \int |p|^2 t_{\gamma_0}(p) \d p + \hbar_p^{-3/2}\int t_{\gamma_0}(p) \int |q|^2 |\hat{f}(q/\sqrt{\hbar_p})|^2 \d q \d p\nonumber\\
			&= \tr\big((-\hbar^2 \Delta) \gamma_0\big) + \hbar_p \tr(\gamma_0)\int |k|^2 |\hat{f}(k)|^2 \d k\notag\\
			&= \tr\big((-\hbar^2 \Delta) \gamma_0\big) + \hbar_p \tr(\gamma_0)\|\nabla f\|_{L^2}^2.
		\end{align}
		On the other hand, note that we do not need the hypothesis $V \to + \infty$ for the proof of the lower bound, so we can replace $V$ by a bounded potential (that we still denote by $V$) as long as it is equal to $V$ on $\supp \rho^{\mathrm{TF}}$ and greater than $\max_{\mathrm{supp}\rho^{\mathrm{TF}}}V$ outside. In particular, we can ask that $V\in W^{1, \infty}$. Thus, we obtain
		\begin{equation}\label{equation_potential_energy_husimi}
			\frac{1}{(2\pi\hbar)^d} \iint V(x) \langle f_{xp}^{\hbar} | \gamma_0 | f_{xp}^{\hbar}\rangle\d x \d p = \int V \big(\rho_{\gamma_0} * |f^{\hbar}|^2\big)= \int \big(V * |f^{\hbar}|^2\big)\rho_{\gamma_0}= \tr(V\gamma_0) + O\big(\|\nabla V\|_{L^{\infty}}\sqrt{\hbar_x}\big).
		\end{equation}
		Therefore, if we choose again $\sqrt{\hbar_x} = \hbar_p = \hbar^{2/3}$, we find
		\begin{equation}\label{equation_variational_energy_lower_bound}
			\tilde{E}^{\mathrm{q}}_{\hbar}(\Lambda)\geq N \tilde{E}^{\mathrm{cl}}(\Lambda) + O(N^{7/9}),
		\end{equation}
		which together with (\ref{equation_variational_energy_upper_bound}) finally gives
		\begin{equation}\label{equation_quantum_semiclassical_variational}
			\tilde{E}^{\mathrm{q}}_{\hbar}(\Lambda) = N \tilde{E}^{\mathrm{cl}}(\Lambda) + O(N^{7/9}).
		\end{equation}
		\textbf{Estimate on the number of particles.} We still denote by $\gamma_0$ the minimizer of the quantum variational problem (\ref{equation_quantum_energy_variational}). Let $\mu \in \R$ (think of it as small), because of \eqref{equation_quantum_semiclassical_variational}, we have
		\begin{equation}
			\mu N^{\mathrm{q}}_{\hbar}(\Lambda) = \mu \tr(\gamma_0) \leq \tilde{E}^{\mathrm{q}}_{\hbar}(\Lambda) - \tilde{E}^{\mathrm{q}}_{\hbar}(\Lambda + \mu) \leq N\tilde{E}^{\mathrm{cl}}(\Lambda) - N\tilde{E}^{\mathrm{cl}}(\Lambda + \mu) + O(N^{7/9}).
		\end{equation}
		We know that $\frac{\d}{\d \Lambda}E^{\mathrm{cl}}(\Lambda) = \Lambda\frac{\d}{\d \Lambda}N^{\mathrm{cl}}(\Lambda)$ \textcolor{red}{mettre référence sur explications}, and hence
		\begin{equation}
			\frac{\d}{\d \Lambda}\tilde{E}^{\mathrm{cl}}(\Lambda) = -N^{\mathrm{cl}}(\Lambda).
		\end{equation}
		Then, as $N^{\mathrm{cl}}$ is differentiable, we find
		\begin{equation}
			\tilde{E}^{\mathrm{cl}}(\Lambda) - \tilde{E}^{\mathrm{cl}}(\Lambda+\mu) = \mu N^{\mathrm{cl}}(\Lambda) + O(\mu^2), 
		\end{equation}
		and thus
		\begin{equation}\label{equation_quantum_semiclassical_number}
			\mu N^{\mathrm{q}}_{\hbar}(\Lambda) \leq \mu N N^{\mathrm{cl}}(\Lambda) + O(N\mu^2) + O(N^{7/9}).
		\end{equation}
		Then, we take $\mu = \pm N^{-1/9}$, we finally have
		\begin{equation}\label{equation_}
			N^{\mathrm{q}}_{\hbar}(\Lambda) = N N^{\mathrm{cl}}(\Lambda) + O(N^{8/9}).
		\end{equation}
		\textbf{Conclusion for the energy.} Equations (\ref{equation_quantum_energy_variational}), (\ref{equation_semiclassical_energy_variational}), (\ref{equation_quantum_semiclassical_variational}) and (\ref{equation_quantum_semiclassical_number}) provide the expected result for the energy.
	\end{proof}
	
	\begin{lemma}\label{lemma_kinetic_energy_low_frequency_husimi}
		Let $0 \leq \gamma \leq 1$ be a trace class operator such that $\tr(\gamma) = \tr(-\hbar^2\Delta \gamma)= O(N)$. We set $m(x, p) = \langle f_{x, p}^{\hbar}|\gamma f_{x, p}^{\hbar}\rangle$. We have
		\begin{equation}
			\tr\big((-\hbar^2 \nabla (1 - \Gamma(p))\nabla\gamma\big) = \frac{N}{(2\pi)^3}\iint |p|^2 (1 - \Gamma(p))m(x, p)\d x\d p + O(N^{5/6})
		\end{equation}
	\end{lemma}
	\begin{proof}
		We have
		\begin{equation}
			\frac{1}{(2\pi \hbar)^3} \int (1 - \Gamma(p))|p|^2 \langle f_{x, p}^{\hbar} | \gamma_0 | f_{x,p}^{\hbar}\rangle\d x \d p = \iint (1 - \Gamma(p-q))|p-q|^2  t_{\gamma_0}(p) |g^{\hbar}(q)|^2 \d q \d p.
		\end{equation}
		Then, we can decompose $|p-q|^2 = |p|^2 - 2p.q + |q|^2$, and compute each integral. First,
		\begin{equation}
			\int (1 - \Gamma(p-q)) |g^{\hbar}(q)|^2 \d q = \int \big(1 - \Gamma(p - \sqrt{\hbar}k)\big)|\hat{f}(k)|^2 \d k = 1 - \Gamma(p) + O(\sqrt{\hbar}).
		\end{equation}
		Then, 
		\begin{equation}
			\iint (1 - \Gamma(p-q))|q|^2  t_{\gamma}(p) |g^{\hbar}(q)|^2 \d q \d p \leq \iint|q|^2  t_{\gamma}(p) |g^{\hbar}(q)|^2 \d q \d p = \hbar\tr(\gamma)\int |k|^2 |\hat{f}(k)|^2\d k.
		\end{equation}
		Finally,
		\begin{equation}
			\left|\int q\big(1 - \Gamma(p - q)\big)|g^{\hbar}(q)|^2 \d q \right| \leq \sqrt{\hbar} \int |k||\hat{f}(k)|^2\d k \leq \frac{\sqrt{\hbar}}{2} \int (1 + |k|^2)|\hat{f}(k)|^2\d k,
		\end{equation}
		and
		\begin{equation}
			\left| \int p t_{\gamma}(p) \d p\right| \leq \frac{1}{2} \int (1 + |p|^2)t_{\gamma}(p)\d p = \frac{1}{2}\Big(\tr(-\hbar^2 \Delta \gamma) + \tr(\gamma)\Big).
		\end{equation}
		Therefore,
		\begin{equation}
			\frac{1}{(2\pi\hbar)^3}\iint |p|^2 (1 - \Gamma(p)) m(x, p)\d x \d p = \tr\big((-\hbar^2 \nabla (1 - \Gamma(p))\nabla\gamma\big) + \Big(\tr(-\hbar^2 \Delta\gamma) + \tr(\gamma)\Big) O(\sqrt
			\hbar).
		\end{equation}
	\end{proof}
	
	\section*{Acknowledgements}
	
	I thank Nicolas Rougerie for many fruitful discussions and for his proofreading of this manuscript. I also thank Laurent Laflèche for some explanations of his joint work with Esteban C\'ardenas \cite{cardenas_commutator_2025}.
	
	\printbibliography

\end{document}